\documentclass[ejsv2,noshowframe,preprint]{imsart}

\RequirePackage{amsthm,amsmath,amsfonts,amssymb}
\RequirePackage[numbers]{natbib}
\RequirePackage{graphicx}
\RequirePackage{tikz}
\usetikzlibrary{arrows, positioning,shapes.geometric,calc,shapes.arrows}
\RequirePackage{xcolor}
\RequirePackage{float}
\RequirePackage{pdflscape}
\RequirePackage[colorlinks,citecolor=blue,linkcolor=blue,urlcolor=blue]{hyperref}

\startlocaldefs

\theoremstyle{plain}
\newtheorem{theorem}{Theorem}[section]      
\newtheorem{proposition}{Proposition}[section]
\newtheorem{corollary}{Corollary}[section]  
\newtheorem{lemma}{Lemma}[section]          
\theoremstyle{definition}
\newtheorem{definition}{Definition}[section]
\newtheorem{assumption}{Assumption}[section]
\newcommand{\indep}{\perp\!\!\!\perp}
\endlocaldefs

\AddToHook{begindocument/end}{\hypersetup{pdfsubject={Preprint}}}

\begin{document}

\begin{frontmatter}

\title{Sensitivity Bounds and Conservative Inference for Contagion under Latent Homophily}
\runtitle{Sensitivity Bounds and Inference for Contagion}

\begin{aug}
\author[A]{\fnms{Duncan~A.}~\snm{Clark}\ead[label=e1]{dac6@williams.edu}}
\address[A]{Williams College\printead[presep={,\ }]{e1}}
\runauthor{D. A. Clark}
\end{aug}

\begin{abstract}
Whether connected units are similar because influence spreads across ties or
because similar units form ties is a long-standing problem.
We study a fixed network with two waves of nodal outcomes. Rather than
positing a parametric model for network formation, we consider identification
of contagion under latent homophily as a selection-bias problem.
We define a focal-component controlled direct effect
(CDE) that holds a tie present, changes one alter's lagged outcome, and permits
dependence on the remaining fixed network background. We show that the gap
between the CDE and the observed connected-dyad risk ratio is governed by how
strongly a latent variable shifts the composition of connected dyads. Under
stated mean-exchangeability and sensitivity restrictions, we develop
interpretable nonparametric bounds for the primary naturally connected target.
For inference conditional on the observed network and baseline outcomes,
heterogeneous dyad-specific means can invalidate the usual
inclusion--exclusion variance estimator; we establish asymptotically
conservative one-sided limits under conditional actor dissociation and stated
regularity conditions. A simulation study characterizes the bounds' error
control and power. We apply the framework to the 2008 U.S. House votes on the
Troubled Asset Relief Program. The connected contrast among legislators
suggests vote contagion and survives mild latent homophily and outcome
susceptibility.
\end{abstract}

\begin{keyword}[class=MSC]
\kwd{62D20}
\kwd{91D30}
\kwd{62P25}
\kwd{62G05}
\end{keyword}

\begin{keyword}
\kwd{contagion}
\kwd{homophily}
\kwd{sensitivity analysis}
\kwd{social networks}
\kwd{selection bias}
\kwd{risk ratio}
\end{keyword}

\end{frontmatter}

\section{Introduction}

Whether connected units are similar because characteristics spread through
ties or because similar units form ties, the contagion-versus-homophily
problem, has long been studied across fields
\cite{Leenders1995,christakis2007,review_pastor_satorras_2015,Franken2023,Cai2025_finance}.
In the social sciences these mechanisms are often called influence and
selection. In the presence of latent homophily they are generically confounded
\cite{shalizi_thomas}; we develop a sensitivity-analysis approach to assessing
contagion under such latent homophily.

We develop sensitivity bounds and one-sided inference for a focal alter-state
effect among naturally connected dyads observed in a single fixed network.
Building on established selection-bias inequalities
\cite{ding2016sensitivity,smith2019bounding}, we specify the intervention,
target population, and mean-exchangeability conditions needed to interpret the
bounds causally. For inference conditional on the observed design, we show how
heterogeneous dyad-specific means can make the usual inclusion--exclusion
variance estimator anti-conservative. Under conditional actor dissociation and
stated regularity conditions, we establish asymptotically conservative
actor-sum inference and corresponding causal lower confidence limits.

Our motivating application is the 2008 U.S.\ House votes authorizing the
controversial \cite{ramirez2012bailout} Troubled Asset Relief Program (TARP).
The bill failed on September 29 and passed in amended form four days later,
with fifty-eight members switching from Nay to Yea
(Section~\ref{sec:tarp-application}). Existing roll-call and network studies
explain voting and tie-based similarity
\cite{mian2010political,ramirez2012bailout,tahoun2019wealth,Fowler2006Connecting,Kirkland2011,NBERw16437}, but not whether votes spread through relationships among legislators. Such relationships can reflect ideology, sub-party allegiances, and constituency pressure, making the TARP vote an example of contagion estimates being confounded by possible latent homophily.

Sensitivity analysis replaces an assumption of no unmeasured confounding with
quantitative restrictions and asks whether the substantive conclusion
survives \cite{Cornfield1959,Bross1966,rosenbaum1987sensitivity,VanderWeeleDing2017,Cinelli2020}.
Most directly, \cite{vanderweele2011contagion} asked how strongly latent
homophily or environmental confounding would have to bias apparent social
contagion effects to explain them away. We adopt that strategy for a specific
dyadic selection-bias problem.

For observed-risk inference, we use a dependency-graph central limit theorem
\cite{janson_1988} in the form developed for dyadic data by
\cite{tabord_meehan_2019}. The actor-sharing analysis allows dependence
between dyads sharing an actor; an outgoing-ego cluster bootstrap provides
an alternative under a stronger independence assumption.

Our primary estimand is the lagged focal-component controlled direct effect
(CDE) among naturally connected dyads, not a contemporaneous, propagated, or
network-wide spillover effect. Under target-specific mean bridges and stated
sensitivity restrictions, it is bounded below by the observed connected-dyad
risk ratio divided by interpretable bias factors. Positive contagion is
identified when the corresponding one-sided lower confidence limit exceeds
one. The forced-contact benchmark uses the same observed-risk inference but
also requires transport from naturally connected to all eligible dyads.

A related nonparametric treatment is given by \cite{vansteeg_2013}. Starting
from the Shalizi--Thomas graph and conditioning on a known directed edge, they
use linear programming and Handelman's representation to lower-bound the
effect of an alter's lagged binary state uniformly over hidden attributes.
Their object is a dyad's action sequence, typically over several time steps,
rather than a two-wave controlled-contact risk ratio; we do not pursue those
sequence-level certificates.

Other approaches use latent-space or stochastic-block assumptions
\citep{McFowland2023}, detailed longitudinal political-network data
\citep{Uppala_Desmarais_2023}, dynamic selection-and-influence models
\citep{steglich2010}, or dynamic matched sampling \citep{Aral2009}. We instead
consider latent homophily, a fixed network, and the two outcome waves minimally
needed for a lagged contagion contrast. Work on causal inference under
interference encodes spillover through exposure mappings
\citep{hudgens2008,ogburn2014,aronow2017,savje2021}, generally taking the
network as known or fixed. Our focus is the distinct selection problem:
unobserved $U$ can change who appears among connected dyads, so the observed
connected contrast need not identify the controlled effect.

The simulations illustrate a robustness--power trade-off: sensitivity
adjustment limits false-positive contagion conclusions but loses power as the
permitted confounding strengthens. In the TARP application, the risk ratio survives mild but not moderate or strong restrictions on latent homophily.

Section~\ref{sec:notation-setup} introduces the fixed-network setting, and
Section~\ref{sec:rr-target-overview} defines the risk-ratio targets.
Section~\ref{sec:rr-cde} develops the bounds, and
Section~\ref{sec:cde-estimation} develops estimation and one-sided inference.
Section~\ref{sec:simulation-study}
evaluates the method in simulations before
Section~\ref{sec:tarp-application} presents the TARP application.

\section{Notation and Setup}\label{sec:notation-setup}

We consider networks with a fixed number of nodes, $n$. Random nodal outcomes are
denoted $Y$, subscripts index the node set and superscripts index time, i.e.
$Y_{i}^{t}$ denotes the outcome on node $i$ at time $t$. Observed nodal
covariates are denoted $X_{i}$ and unobserved fixed nodal covariates are
denoted $U_{i}$. Binary edges $A_{ij}$ yield realizations of network
adjacency matrices $A \in \lbrace 0,1 \rbrace^{n \times n}$. Throughout,
pre-exposure measured information is absorbed into a generic conditioning
set $c$. For dyad-level arguments, we write $U_{i,j}=(U_i,U_j)$ for the latent
homophily tuple associated with dyad $(i,j)$. For dyad $(i,j)$ we refer to $i$ as the ego and $j$ as the alter.

\subsection{Fixed networks with two time steps}\label{sec:fe-scope}

We consider a network which forms under latent homophily, with a subsequent two-time-step process for the outcome variable $Y$. Figure \ref{fig:ST_DAG} shows a directed acyclic graph (DAG) \cite{pearl2009} that represents the dependence structure, only the latent $U$ variables are shown explicitly. Two time steps for the outcomes are the minimal setting in which contagion can be detected. We do not consider the case with more longitudinal data, as the adjacency matrix $A$ would need to be dynamic, which is beyond the scope of this paper. The fixed network case with more time steps simply increases the size of the dyadic data under a suitable time step exchangeability assumption.

\begin{figure}[ht]
  \centering
\begin{tikzpicture}[
    >=latex,             
    every node/.style={font=\small, draw, inner sep=2pt},
    endo/.style={rectangle},
    medi/.style={rectangle},
    ell/.style={shape=ellipse},
    node distance=1.5cm and 2.5cm
  ]

  \node[ell](Ui)   {$U_{i}$};
  \node[ell, right=of Ui](Uj)   {$U_{j}$};

  \node[endo, below=of Ui]      (Yit1) {$Y_{i}^{t-1}$};
  \node[medi] (Aij) at ($ (Ui)!0.5!(Uj) + (0,-1 cm) $) {$A_{ij}$};
  \node[endo, below=of Uj]      (Yjt1) {$Y_{j}^{t-1}$};

  \node[endo, below=4cm of Ui] (Yit)  {$Y_{i}^{t}$};
  \node[endo, below=4cm of Uj]  (Yjt)  {$Y_{j}^{t}$};


  \draw[->] (Ui)   -- (Aij);
  \draw[->] (Ui)  -- (Yit1);
  \draw[->] (Ui)   to[bend right=30,out=315]  (Yit);
  \draw[->] (Uj)   to[bend left=30,out=45]  (Yjt);

  \draw[->] (Yit1) -- (Yit);

  \draw[->] (Aij)  -- (Yit);
  \draw[->] (Aij)  -- (Yjt);

  \draw[->] (Uj)   -- (Aij);
  \draw[->] (Uj)   -- (Yjt1);

  \draw[->] (Yjt1) -- (Yjt);

  \draw[->] (Yjt1) -- (Yit);
  \draw[->] (Yit1) --(Yjt);

\end{tikzpicture}
\caption{DAG showing latent \(U\) variables, with contagion and latent homophily for a single dyad; observed fixed covariates \(X\) are suppressed.}
\label{fig:ST_DAG}
\end{figure}
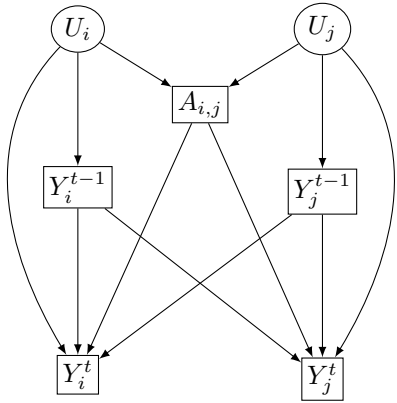

We consider data generated from Figure \ref{fig:ST_DAG}, with binary nodal
outcomes $Y_{i}$ and categorical unobserved fixed nodal covariates $U_{i}$. The edges $A$ are
considered to have formed in the presence of latent homophily on the $U$
variables. The observed $X$ variables are treated as known and absorbed
into the adjustment set $c$ below. The edges are assumed to have formed prior
to the evolution of the $Y$ variable on top of the network.

Contagion is confounded with latent homophily \cite{shalizi_thomas}, we can see this from the DAG. The obvious approach to estimate contagion would be to model the effect (causal or otherwise) of $Y_{j}^{t-1}$ on $Y_{i}^{t}$, conditional on $A_{ij}= 1$. This would condition on a collider in the DAG which is well known to induce spurious dependence between $Y_{i}^{t}$ and $Y_{j}^{t-1}$
\cite{pearl1988}, which one could mistake for contagion.  If we were to instead not condition on $A_{ij}= 1$, we would leave a backdoor path open, through the latent homophily inducing variables, again inducing spurious dependence. This motivates the need for our sensitivity analysis.

Absent DAG arrows are also important; in this framework, once the network is formed under latent homophily, the $Y$ nodal variables are assumed not to impact the edges. This
works well where the edges are considered fixed prior to nodal process $Y^{t}$ evolving, but if the outcome informs non fixed edges in future time steps, it is misspecified. See \cite{clark_handcock_2024} for an example DAG where the network itself evolves in time with nodal covariates, with few restrictions.

Two observed waves do not yield a routine within-ego or within-dyad
fixed-effects solution. A dyad intercept would require
repeated transitions on the same pair. The confounding of interest is also
not only a level shift in $Y$. It is selection into the connected-dyad
sample: conditioning on $A_{ij}=1$ opens a collider path through $U$
\citep{shalizi_thomas}. Absorbing an ego intercept, or controlling for lagged
outcomes, does not close that path, because a fixed latent type can still
shift current outcomes after the lags
\citep{shalizi_thomas,vanderweele2011contagion}. The sensitivity analysis
targets that residual composition shift among connected dyads.

The latent type $U$ is fixed over the two waves. Time-varying unmeasured
confounding, such as evolving ideology or a latent factor that
changes between waves and jointly moves outcomes and the interpretation of
existing ties, is outside the present bounds. If such a factor is measured
before the focal exposure, it can be placed in $c$. If it is unmeasured, the
sensitivity parameters on a fixed $U$ do not cover it.

\section{Risk-Ratio Targets}\label{sec:rr-target-overview}

Connected dyads are usually statistically different from unconnected dyads, in social network settings \cite{McFowland2023}. As social networks are typically sparse, the presence of an edge is almost always highly informative of the social process that created it. Thus, if behaviors
are strongly informed by the same social process, estimation on just connected dyads leads to strong selection bias.

We follow the logic in \cite{smith2019bounding}, and target risk ratios considering
dyads $(i,j)$, and framing our problem as selection bias. Let the pre-exposure measured information we
adjust for be collected in $c$, for example $c = \lbrace X_i, X_j, Y_{i}^{t-1}
\rbrace$; see Section~\ref{sec:fe-scope}. 

We make the key distinction between causal risk ratios, denoted CRR, and
observed or realized-network risk ratios, denoted RR. CRRs require an interventional distribution 
to be considered, whereas RRs rely only on the natural distribution. We formalize the notion of 
contagion with the following definitions. The potential-outcomes notation used
below can equivalently be read in do-notation, with the
interventions placed on the conditioning side with the do operator.

\begin{definition}[Unconditional causal risk ratio]\label{def:rr-uncond-causal}
For the alter-state treatment $Z_{i,j}=Y_j^{t-1}$, define
\begin{align*}
\text{CRR}^{\text{uncond}}
&=
\frac{
P(Y_i^t(Z_{i,j}=1)=1\mid c)
}{
P(Y_i^t(Z_{i,j}=0)=1\mid c)
}.
\end{align*}
\end{definition}

The observed-data counterpart is
\begin{definition}[Unconditional observed risk ratio]\label{def:rr-uncond-observed}
\begin{align*}
\text{RR}^{\text{uncond}}
&=
\frac{
P(Y_i^t=1\mid Z_{i,j}=1,c)
}{
P(Y_i^t=1\mid Z_{i,j}=0,c)
}.
\end{align*}
\end{definition}

In a non-networked setting, Definition \ref{def:rr-uncond-causal} is
nonparametrically identified by Definition \ref{def:rr-uncond-observed} under
the usual assumptions; see, for example, \cite{hernan_robins_2020}.
\begin{enumerate}
    \item \textbf{Consistency/focal-intervention specification:} if
          $Z_{i,j}=z$, then $Y_{i}^{t} = Y_{i}^{t}(Z_{i,j}=z)$ for
          $z\in\{0,1\}$; for the network interpretation below, the rest of the
          realized network is held fixed rather than assumed to have no effect
          on the ego.
    \item \textbf{Positivity:} $P(Z_{i,j}=z \mid c) > 0$ for
          $z\in\{0,1\}$ whenever $P(c)>0$.
    \item \textbf{Conditional exchangeability:}
          $Y_{i}^{t}(Z_{i,j}=z) \indep Z_{i,j} \mid c$ for $z\in\{0,1\}$.
\end{enumerate}

In a network setting, the conditional exchangeability assumption is violated in the presence of latent homophily, as discussed in Section \ref{sec:notation-setup} as the outcomes $Y_{i}^{t}$ and $Y_{j}^{t-1}$ are not independent given $c$.
The $CRR^{\text{uncond}}$  is best understood as broad alter-state dependence, though is typically close to 1 given the sparsity of typical social networks.
It averages over connected and unconnected dyads and can be affected by
contagion, homophily, and shared social structure, among other complexities of the social process
that generated the data. It is therefore a useful orientation point,
but should not be interpreted as a contagion estimand.

One might next consider an intervention that sets both the alter state and the edge to be present.

\begin{definition}[Focal-component intervention]\label{def:focal-component}
Throughout, \emph{focal} means dyadwise: the intervention is defined for one
ordered pair $(i,j)$ at a time. For that pair, the focal components are the
alter's lagged state $Y_j^{t-1}$ and the tie $A_{ij}$. The nonfocal background
is
\[
H_{ij}:=\left(A_{-(ij)},\mathbf Y_{-j}^{t-1},\mathbf X\right),
\]
where $A_{-(ij)}$ contains all edges other than the $i$--$j$ tie,
$\mathbf Y_{-j}^{t-1}$ contains all lagged outcomes other than $Y_j^{t-1}$, and
$\mathbf X$ contains measured pre-exposure information. The potential outcome
$Y_i^t(y,1;H_{ij})$ sets $Y_j^{t-1}=y$ and $A_{ij}=1$ while leaving $H_{ij}$ at
its realized value. The intervention is not simultaneous or network-wide.
\end{definition}

While we use the focal-component intervention to define contagion, we do not require independence between dyads to estimate the risk ratios. Formal independence assumptions are required for inference and are developed in Section~\ref{sec:cde-estimation}.

\begin{definition}[All-dyad composite causal and observed risk ratios]\label{def:rr-composite}
Let $Z_{i,j}^{A}=Y_j^{t-1}\cdot A_{ij}$ denote the composite exposure requiring
both an altered prior state and a tie. The all-dyad composite causal risk ratio
is
\begin{align*}
\text{CRR}^{\text{target}}_{\text{all}}
&=
\frac{
P(Y_i^t(Z_{i,j}^{A}=1)=1\mid c)
}{
P(Y_i^t(Z_{i,j}^{A}=0)=1\mid c)
}.
\end{align*}
The realized-network observed counterpart is
\begin{align*}
\text{RR}^{\text{target}}_{\text{all}}
&=
\frac{
P(Y_i^t=1\mid Z_{i,j}^{A}=1,c)
}{
P(Y_i^t=1\mid Z_{i,j}^{A}=0,c)
}.
\end{align*}
\end{definition}

These ratios are closer to the empirical language of exposure because
the exposed state requires an affected connected alter. It is still not a clean
causal contagion target. First, the intervention $Z_{i,j}^{A}=0$ is a composite
intervention: it can correspond to an unaffected connected alter, an affected
unconnected alter, or an unaffected unconnected alter. Second, even the
observed counterpart inherits homophily from the edge-formation process because
the realized edge $A_{ij}$ is part of the exposure definition. Thus
$\text{RR}^{\text{target}}_{\text{all}}$ can reflect contagion, but it can also
reflect sorting in the $Y$ outcome due to latent homophily. Again, $\text{RR}^{\text{target}}_{\text{all}}$ does not identify $CRR^{\text{target}}_{\text{all}}$ due to the conditional exchangeability assumption being violated.

\begin{definition}[Forced-contact CDE causal risk ratio]\label{def:rr-cde}
The forced-contact controlled direct effect (CDE) causal risk ratio compares the
effect of the alter's prior outcome while setting the focal edge present for the
full dyad population:
\begin{align*}
\text{CRR}^{\text{fc}}_{\text{CDE}}
&:=
\frac{
P\{Y_i^t(1,1;H_{ij})=1\mid c\}
}{
P\{Y_i^t(0,1;H_{ij})=1\mid c\}
}.
\end{align*}
This is the stronger causal contagion target because it asks what would happen
if contact were set present for dyads whether or not they naturally formed ties.
That extrapolation requires a much stronger, and in most social-network
settings difficult to justify, transport condition: within latent type,
connected risks must remain informative after forcing ties that did not form.
\end{definition}

\begin{definition}[Naturally connected CDE]\label{def:rr-cde-conn}
The naturally connected CDE causal risk ratio compares the same controlled
alter-state interventions, but among dyads that naturally form ties:
\begin{align*}
\text{CRR}^{\text{conn}}_{\text{CDE}}
&:=
\frac{
P\{Y_i^t(1,1;H_{ij})=1\mid A_{ij}=1,c\}
}{
P\{Y_i^t(0,1;H_{ij})=1\mid A_{ij}=1,c\}
}.
\end{align*}
This target asks for the dyad-indexed alter-state effect among dyads that
naturally form ties. It is weaker than the forced-contact CDE because it does not
extrapolate to all possible dyads; it conditions the target population on
natural tie formation.
\end{definition}

For both CDE targets, potential outcomes are defined focal dyad by focal dyad.
The two risk-ratio arms average over the same focal-background distribution;
only the focal alter state changes. The response may depend arbitrarily on
$H_{ij}$, including all neighbors, degree, aggregate exposures, and interactions
among them. Thus ``dyadic'' identifies the component changed and the target. The
intervention is not simultaneous or network-wide.
We interpret this dyad-indexed alter-state effect with the focal tie held
present as a broad definition of contagion.

Mappings such as the number or proportion of exposed neighbors instead suppress
alter identity and target a different contrast. Our contribution is therefore
not greater representational generality than exposure mappings \cite{aronow2017}. It is a
selection-bias sensitivity analysis for a focal tie and focal alter state when
latent homophily changes dyad composition. The target does not identify a
contemporaneous effect, a propagated intervention through later third-party
outcomes, a network-wide policy, or a path-specific transmission mechanism.

As network edges are often so strongly suggestive of social selection processes, forcing
an edge into presence can be quite different from the natural process of tie formation. Thus, targeting
$\text{CRR}^{\text{fc}}_{\text{CDE}}$ is expected to require challenging assumptions
about how the interventional dyad distributions compare to the natural distribution.
$\text{CRR}^{\text{conn}}_{\text{CDE}}$ is expected to be easier to target, and more interpretable
as the contagion signal among edges that plausibly formed naturally, potentially under latent homophily.

\begin{definition}[Observed connected-dyad risk ratio]\label{def:rr-obs-conn}
In observed network data, we do not observe dyads under the controlled
alter-state and contact interventions. We observe dyads for which $A_{ij}=1$
naturally, giving the connected-dyad risk ratio
\begin{align*}
\text{RR}^{\text{obs}}_{\text{conn}}
&=
\frac{
P(Y_i^t=1\mid Y_j^{t-1}=1,A_{ij}=1,c)
}{
P(Y_i^t=1\mid Y_j^{t-1}=0,A_{ij}=1,c)
}.
\end{align*}
\end{definition}

These targets form the paper's main hierarchy. $\text{RR}^{\text{uncond}}$ is a
broad alter dependence contrast.
$\text{CRR}^{\text{target}}_{\text{all}}$ and
$\text{RR}^{\text{target}}_{\text{all}}$ describe causal and observed versions
of a composite network-exposure contrast. The CDE targets then isolate
contagion by comparing alter states while holding the edge present.
$\text{CRR}^{\text{conn}}_{\text{CDE}}$ is the naturally connected target among
dyads that form ties, while
$\text{CRR}^{\text{fc}}_{\text{CDE}}$ is the stronger forced-contact target,
defined for the full dyad population. The empirical contrast available from a
realized network is $\text{RR}^{\text{obs}}_{\text{conn}}$. The sensitivity
analysis below asks how far this observed connected-dyad contrast may be from
each CDE target.  Table~\ref{tab:rr-hierarchy} collects these risk ratios for reference. Each
causal target is paired with its observed-data counterpart.

\begin{table}[H]
\centering
\caption{The risk-ratio targets introduced in this section, with the
definition that states each. ``Causal'' targets are interventional quantities; ``Observed'' contrasts use only the natural distribution and are what may be estimated from the data.}
\label{tab:rr-hierarchy}
\begin{tabular}{llp{0.56\linewidth}}
\hline
Risk ratio & Form & Description and relationship \\
\hline
\hline
$\text{CRR}^{\text{uncond}}$ & Causal & Broadest contrast;
mixes contagion, homophily, and shared structure, so it is an orientation point
rather than a contagion estimand. \\
$\text{RR}^{\text{uncond}}$ & Observed & Observed-data counterpart of
$\text{CRR}^{\text{uncond}}$ \\
\hline
$\text{CRR}^{\text{target}}_{\text{all}}$ & Causal & Composite exposure
requiring both an altered prior state and a tie, over all dyads
(Def.~\ref{def:rr-composite}). Closer to ``having an affected contact,'' but its
unexposed arm pools several distinct states. \\
$\text{RR}^{\text{target}}_{\text{all}}$ & Observed & Observed counterpart of
$\text{CRR}^{\text{target}}_{\text{all}}$. \\
\hline
$\text{CRR}^{\text{conn}}_{\text{CDE}}$ & Causal & Controlled direct effect of
the alter's prior state holding the tie present, standardized to dyads that
naturally form ties (Def.~\ref{def:rr-cde-conn}). \\
$\text{CRR}^{\text{fc}}_{\text{CDE}}$ & Causal & The same controlled direct
effect standardized to the full dyad population, as if contact were
forced (Def.~\ref{def:rr-cde}). Stronger than the connected target and requires
the extra selection-into-ties extrapolation. \\
\hline
$\text{RR}^{\text{obs}}_{\text{conn}}$ & Observed & Observed connected-dyad risk
ratio among naturally tied dyads (Def.~\ref{def:rr-obs-conn}). The obvious contrast
estimable from a fixed two-wave network and the empirical input to the bounds;
Section~\ref{sec:rr-cde} bounds how far it can lie from
$\text{CRR}^{\text{conn}}_{\text{CDE}}$ and $\text{CRR}^{\text{fc}}_{\text{CDE}}$. \\
\hline
\end{tabular}
\end{table}

\section{CDE Sensitivity Bounds}\label{sec:rr-cde}

We now focus on CDE causal risk ratios. Let $U_{i,j}=(U_i,U_j)$ denote the latent dyad
type. For $y\in\{0,1\}$, define the within-connected latent-stratum risk
\[
r_y^{\text{conn}}(u)
:=
P\{Y_i^t(y,1;H_{ij})=1\mid A_{ij}=1,U_{i,j}=u,c\}.
\]
This is an interventional risk: it asks what the ego's outcome risk would be if
the alter's prior state were set to $y$ and the edge were held present, among
naturally connected dyads with latent type $U_{i,j}=u$.

The observed analogue within naturally connected dyads is
\[
r_y^{\text{obs}}(u)
:=
P(Y_i^t=1\mid Y_j^{t-1}=y,A_{ij}=1,U=u,c).
\]
Using
\[
p_y(u)=P(U=u\mid Y_j^{t-1}=y,A_{ij}=1,c), \qquad y\in\{0,1\},
\]
the observed connected-dyad contrast in Definition \ref{def:rr-obs-conn} is
\[
\text{RR}^{\text{obs}}_{\text{conn}}
=
\frac{
\sum_u r_1^{\text{obs}}(u) p_1(u)
}{
\sum_u r_0^{\text{obs}}(u) p_0(u)
}.
\]

\begin{assumption}[Connected consistency and positivity]\label{assump:causal-bridge}
For each $y\in\{0,1\}$:
\begin{enumerate}
    \item \textbf{Consistency:} for dyads with $Y_j^{t-1}=y$ and
          $A_{ij}=1$, the observed outcome equals the controlled-contact
          potential outcome, $Y_i^t=Y_i^t(y,1;H_{ij})$.
    \item \textbf{Positivity:} such dyads occur with positive probability within
          levels of $(U,c)$ whenever the latent stratum is part of the target
          population.
\end{enumerate}
\end{assumption}

\begin{assumption}[Within-connected mean exchangeability]\label{assump:background-mean-bridge}
For $y\in\{0,1\}$ and every relevant $(u,c)$,
\[
E\{Y_i^t(y,1;H_{ij})\mid
Y_j^{t-1}=y,A_{ij}=1,U_{i,j}=u,c\}
=r_y^{\mathrm{conn}}(u).
\]
\end{assumption}
Any attribute that both predicts the controlled-contact mean and is imbalanced
across natural alter states, including group or third-party features, must be
represented in $U$ or $c$. Residual mean imbalance through $H_{ij}$ is outside
the sensitivity analysis.

Under Assumptions~\ref{assump:causal-bridge} and
\ref{assump:background-mean-bridge}, the observed
connected-stratum risks equal the controlled-contact risks:
\[
r_y^{\text{obs}}(u)=r_y^{\text{conn}}(u)
\qquad\text{for all }y\in\{0,1\}\text{ and all }u.
\]
Thus the observed connected-dyad contrast can be rewritten as
\[
\text{RR}^{\text{obs}}_{\text{conn}}
=
\frac{
\sum_u r_1^{\text{conn}}(u) p_1(u)
}{
\sum_u r_0^{\text{conn}}(u) p_0(u)
},
\]
so the sensitivity parameters below need only address latent-type composition
shifts. The two CDE targets from Section~\ref{sec:rr-target-overview} can now be
written as latent-standardized contrasts.
When several measured strata are retained, the analyst specifies one value for
each sensitivity parameter below. Each value is interpreted as an upper bound
that holds within every retained $c$; we do not consider stratum-specific
sensitivity specifications.

\subsection{Naturally connected CDE target.}
The first target asks for the CDE standardized to the latent distribution among
naturally connected dyads. Let
\[
p_A(u)=P(U_{i,j}=u\mid A_{ij}=1,c).
\]
Define
\[
\text{CRR}^{\text{conn}}_{\text{CDE}}
:=
\frac{
\sum_u r_1^{\text{conn}}(u)p_A(u)
}{
\sum_u r_0^{\text{conn}}(u)p_A(u)
}.
\]
This target does not extrapolate to all possible dyads. It asks for contagion
among dyads that naturally form ties. This target choice parallels inference
in the selected population in \cite[Result~5A]{smith2019bounding}, with
naturally connected dyads constituting the selected population. Its bias
relative to the observed connected-dyad contrast is
\begin{align*}
\frac{\text{RR}^{\text{obs}}_{\text{conn}}}
{\text{CRR}^{\text{conn}}_{\text{CDE}}}
&=
\left[
\frac{\sum_u r_1^{\text{conn}}(u)p_1(u)}
{\sum_u r_1^{\text{conn}}(u)p_A(u)}
\right]
\left[
\frac{\sum_u r_0^{\text{conn}}(u)p_A(u)}
{\sum_u r_0^{\text{conn}}(u)p_0(u)}
\right].
\end{align*}
The two bracketed terms show the numerator and denominator bias. They compare
the latent distribution among exposed connected dyads to the natural connected
target distribution, and the natural connected target distribution to the latent
distribution among unexposed connected dyads.

For $y\in\{0,1\}$, define the outcome-risk variation
\[
R_y^{\text{CDE}}
:=
\frac{\max_u r_y^{\text{conn}}(u)}{\min_u r_y^{\text{conn}}(u)}.
\]
Define the connected-target latent distribution shifts
\[
R_{U,1}^{\text{conn}}
:=
\max_u \frac{p_1(u)}{p_A(u)},
\qquad
R_{U,0}^{\text{conn}}
:=
\max_u \frac{p_A(u)}{p_0(u)}.
\]
All of the bounds below are built from a single bounding-factor inequality for a
ratio of weighted-average risks. We state it once and reuse it. Define
\[
B(a,b)=\frac{ab}{a+b-1}.
\]

\begin{lemma}[Bounding factor for a standardized risk ratio]\label{lem:bounding-factor}
Let $r:\mathcal{U}\to(0,1]$ and let $p$ and $q$ be probability distributions on a
common support $\mathcal{U}$. Write $R_Y=\max_u r(u)/\min_u r(u)$ for the
outcome-risk variation and $R_U=\max_u p(u)/q(u)$ for the maximal
distribution shift. Then
\[
\frac{\sum_u r(u)p(u)}{\sum_u r(u)q(u)}
\leq
B(R_Y,R_U)
=
\frac{R_Y R_U}{R_Y+R_U-1}.
\]
The constant is sharp in the closure of the common-support model: two-point
distributions can approach equality arbitrarily closely. 
\end{lemma}

Lemma \ref{lem:bounding-factor} is the Ding--VanderWeele bounding factor
\cite{ding2016sensitivity} as applied to selection problems by \cite{smith2019bounding};
all proofs are deferred to Appendix \ref{app:proofs}. Define the
connected-target bias factor
\[
BF_{\text{CDE,conn}}
:=
B(R_1^{\text{CDE}},R_{U,1}^{\text{conn}})
B(R_0^{\text{CDE}},R_{U,0}^{\text{conn}}).
\]

\begin{proposition}[Naturally connected CDE lower bound]\label{prop:conn-bound}
Under Assumptions~\ref{assump:causal-bridge} and
\ref{assump:background-mean-bridge}, and common latent support of
$p_1$, $p_0$, and $p_A$,
\[
\frac{\text{RR}^{\text{obs}}_{\text{conn}}}{BF_{\text{CDE,conn}}}
\leq
\text{CRR}^{\text{conn}}_{\text{CDE}}.
\]
\end{proposition}

The sensitivity parameters explain the imbalance between exposed connected dyads and
unexposed connected dyads, among the naturally connected target distribution.
Analogous upper bounds can be derived using reverse latent-composition ratios,
but we do not pursue them because identifying positive contagion requires only
that a lower bound for $\text{CRR}^{\text{conn}}_{\text{CDE}}$ exceed one.

\subsection{Forced-contact CDE target.}
The second target asks for the CDE if the contact were set present over the full
dyad population. Define
\[
r_y^{\text{fc}}(u)
:=
P\{Y_i^t(y,1;H_{ij})=1\mid U_{i,j}=u,c\},
\qquad
p(u)=P(U_{i,j}=u\mid c).
\]

\begin{assumption}[Forced-contact mean transportability]\label{assump:forced-contact-transport}
For $y\in\{0,1\}$ and every relevant $(u,c)$,
\[
E\{Y_i^t(y,1;H_{ij})\mid A_{ij}=1,U_{i,j}=u,c\}
=
E\{Y_i^t(y,1;H_{ij})\mid U_{i,j}=u,c\}.
\]
\end{assumption}
Assumption~\ref{assump:background-mean-bridge} bridges alter-state groups
within connected dyads; Assumption~\ref{assump:forced-contact-transport}
transports the resulting risks to all dyads.
Equivalently, $r_y^{\text{conn}}(u)=r_y^{\text{fc}}(u)$. The selection
parameters below control only changes in latent-type composition. Under this
assumption,
$R_y^{\text{CDE}}$ is also the outcome-risk variation of
$r_y^{\text{fc}}(u)$.

Appendix~\ref{app:fc-compositional} gives the latent-standardized form of the
forced-contact target, its bias decomposition relative to the observed
connected-dyad contrast, and the definition of the one-step direct factor
$BF_{\text{CDE,fc}}^{\text{dir}}$. The derivation is analogous to that for the naturally connected target; the resulting bound is as follows.

\begin{proposition}[Direct forced-contact CDE lower bound]\label{prop:direct-fc-bound}
Under Assumptions~\ref{assump:causal-bridge},
\ref{assump:background-mean-bridge}, and
\ref{assump:forced-contact-transport}, and common latent support of
$p_1$, $p_0$, and $p$,
\[
\frac{\text{RR}^{\text{obs}}_{\text{conn}}}{BF_{\text{CDE,fc}}^{\text{dir}}}
\leq
\text{CRR}^{\text{fc}}_{\text{CDE}}.
\]
\end{proposition}

The one-step forced-contact factor is the sharpest of the proposed
forced-contact adjustments, but it bundles imbalance among connected dyads
with selection from all eligible dyads into naturally observed ties. This can
make its sensitivity parameters difficult to calibrate from substantive
knowledge. Appendix~\ref{app:fc-compositional} therefore develops a
calibration-oriented compositional analysis in substantial detail. It factors
the discrepancy through the naturally connected target, defines separate
selection-into-ties parameters and the composed factor
$BF_{\text{CDE,fc}}$, and establishes the valid lower bound and factor ordering
in Proposition~\ref{prop:composed-fc-bound}. The forced-contact simulation and
application calculations use this composed factor; the primary naturally
connected analysis uses $BF_{\text{CDE,conn}}$ and does not require the
additional selection component.

We retain the forced contact target for completeness, though note that in a network setting the assumptions required to bound it are often extreme, and recommend focussing on the naturally connected target.

\subsection{Ego-centric bounds}\label{sec:ego-centric-bounds}

$BF_{\text{CDE,conn}}$ treats the dyad type
$U=(U_i,U_j)$ as a single sensitivity variable, so its worst case may shift both
the ego type $E=U_i$ and the alter type $V=U_j$ at once. Writing the connected
distribution shift as a product of an ego shift $R_{E,y}$ and an alter shift
$R_{V,y}$ (Appendix~\ref{app:ego-construction}), an ego-centric refinement holds
the ego type fixed and charges only the residual alter-type imbalance within ego
type.

\begin{assumption}[No latent ego-type imbalance]\label{assump:common-ego-type-bias}
Conditional on $c$, the ego latent type $E=U_i$ has a common
distribution across the connected exposed, connected unexposed, and naturally
connected target populations; equivalently, the ego shifts are null,
$R_{E,1}=R_{E,0}=1$. Any residual latent imbalance between these populations then
operates through the alter type $V=U_j$.
\end{assumption}

Assumption~\ref{assump:common-ego-type-bias} restricts the relative strength of
the latent confounding rather than the causal structure: it asserts that the
latent imbalance is alter-carried. The relevant object is the ego-type
composition $P(U_i\mid Y_j^{t-1},A_{ij}=1,c)$ across alter-exposure groups.
Whether $U_i$ still affects $Y_i^t$ after conditioning on the ego's first
outcome is a statement about outcome-risk heterogeneity, not about this
composition. The assumption is most plausible when observed $c$ already
balances ego types across exposed and unexposed connected dyads. It remains
hard to verify directly when $U$ is unobserved; with an observed proxy one
should check the composition directly, and use non-unit $R_{E,y}$ if the
shares only approximately agree.

With the ego shift set to one, the worst-case comparison ranges only over alter
types within ego type. Working within each ego type with the within-ego
outcome-risk variation $R_y^{\text{CDE}}(e)$ and connected-target alter shifts
$R_{V,y}^{\text{conn}}(e)$, and taking the worst case over ego types, define the
ego-centric connected factor $BF_{\text{CDE,conn}}^{\text{ego}}$
(Appendix~\ref{app:ego-construction}).

\begin{proposition}[Ego-centric connected CDE lower bound]\label{prop:ego-bounds}
Under Assumptions~\ref{assump:causal-bridge},
\ref{assump:background-mean-bridge}, and
\ref{assump:common-ego-type-bias}, the ego-centric factor yields the
connected-target lower bound
\[
\text{CRR}^{\text{conn}}_{\text{CDE}}
\geq
\frac{\text{RR}^{\text{obs}}_{\text{conn}}}{BF_{\text{CDE,conn}}^{\text{ego}}}.
\]
\end{proposition}

\section{Estimation and One-Sided Inference}\label{sec:cde-estimation}

The bounds of Section~\ref{sec:rr-cde} are population statements:
they relate $\text{RR}^{\text{obs}}_{\text{conn}}$ to
$\text{CRR}^{\text{conn}}_{\text{CDE}}$ or
$\text{CRR}^{\text{fc}}_{\text{CDE}}$ under
Assumptions~\ref{assump:causal-bridge},
\ref{assump:background-mean-bridge}, and, for forced contact,
\ref{assump:forced-contact-transport}.
This section estimates that observed contrast and then gives two
one-sided procedures, one based on an ego cluster bootstrap and one based on a dyadic dependence central limit theorem. They consider slightly different causal targets, and therefore
use different restatements of those identifying assumptions.

Point estimation and the ego cluster bootstrap are
superpopulation statements. They transfer coverage from
$\text{RR}^{\text{obs}}_{\text{conn}}$ to the population CDE targets
under the assumptions of Section~\ref{sec:rr-cde}; the bootstrap does
not use the finite-frame conditions below.

The primary actor-sharing limits are instead design-based.
They condition on the complete pre-outcome design $\mathcal F_n$
defined below, so that the remaining randomness is the second-wave
outcomes. This is necessary in order to develop asymptotic inference.

Conditional on that frame, the estimator is centered at a
realized-frame observed ratio
$\text{RR}^{\mathrm{obs}}_n(\mathcal F_n)$, not at the population
functional of Section~\ref{sec:rr-cde}.
A confidence limit for that realized-frame contrast transfers to a
causal ratio on the same frame only through finite-frame analogues of
the mean-bridge and transport conditions, stated as
Assumptions~\ref{assump:finite-frame-bridge} and
\ref{assump:finite-frame-fc-transport}.
Those analogues are not implied by the population assumptions:
equality of conditional expectations does not imply equality of
realized averages of dyad-specific potential-outcome probabilities.
We first construct the common-weight estimator, then the actor-sharing
limits, then the ego-cluster bootstrap.

Practically, if the ego bootstrap assumption is plausible, the bootstrap method should be used and interpreted in a superpopulation sense. If not, and the data is large, then the actor-sharing limits should be used and interpreted conditional on the observed pre-outcome design.

\subsection{Point estimation with measured adjustment}

The empirical input to the sensitivity analysis is the observed connected-dyad
risk ratio $\text{RR}^{\text{obs}}_{\text{conn}}$. The rows remain directed
throughout: $(i,j,t)$ records ego $i$ and alter $j$, and the reciprocal row
$(j,i,t)$ is a distinct observation. Without measured strata, define
\[
\widehat r_y^{\text{obs}}
:=
\frac{
\sum_{i\neq j,t} I(Y_i^t=1,Y_j^{t-1}=y,A_{ij}=1)
}{
\sum_{i\neq j,t} I(Y_j^{t-1}=y,A_{ij}=1)
},
\qquad y\in\{0,1\}.
\]
The empirical connected-dyad risk ratio is
\[
\widehat{\text{RR}}^{\text{obs}}_{\text{conn}}
=
\frac{\widehat r_1^{\text{obs}}}{\widehat r_0^{\text{obs}}}.
\]
For a finite measured adjustment set $\mathcal C$, let
\[
N_c=\sum_{i\ne j,t}I(A_{ij}=1,c_{ijt}=c),\qquad
N_{yc}=\sum_{i\ne j,t}I(A_{ij}=1,Y_j^{t-1}=y,c_{ijt}=c),
\]
and let $\widehat r_{yc}^{\text{obs}}$ be the corresponding within-$(y,c)$
outcome mean. The common-weight estimator used below is
\[
\widehat w_c=\frac{N_c}{\sum_{c'\in\mathcal C}N_{c'}},\qquad
\widehat R_y=\sum_{c\in\mathcal C}\widehat w_c
\widehat r_{yc}^{\text{obs}},\qquad
\widehat{\text{RR}}^{\text{obs}}_{\text{conn}}
=\frac{\widehat R_1}{\widehat R_0}.
\]
At the population level, write
\[
w_c=P(c\mid A_{ij}=1),\quad
\mu_{yc}^{\text{obs}}
=P(Y_i^t=1\mid Y_j^{t-1}=y,A_{ij}=1,c),
\]
and
\[
\mu_{yc}^{\text{conn}}
=P\{Y_i^t(y,1;H_{ij})=1\mid A_{ij}=1,c\},
\qquad
\mu_{yc}^{\text{fc}}
=P\{Y_i^t(y,1;H_{ij})=1\mid c\}.
\]
When $\mathcal C$ contains several strata, we use the same notation as above
for the corresponding common-weight extensions:
\[
\text{RR}^{\text{obs}}_{\text{conn}}
=
\frac{\sum_cw_c\mu_{1c}^{\text{obs}}}
{\sum_cw_c\mu_{0c}^{\text{obs}}},
\qquad
\text{CRR}^{\text{conn}}_{\text{CDE}}
=
\frac{\sum_cw_c\mu_{1c}^{\text{conn}}}
{\sum_cw_c\mu_{0c}^{\text{conn}}},
\qquad
\text{CRR}^{\text{fc}}_{\text{CDE}}
=
\frac{\sum_cw_c\mu_{1c}^{\text{fc}}}
{\sum_cw_c\mu_{0c}^{\text{fc}}}.
\]
The connected target averages connected dyads within each $c$; the
forced-contact target averages all eligible dyads within each $c$. Both use the
same connected-frame weights $w_c$, so the observed estimator and both causal
targets share one measured-stratum standardization. It is dyad-weighted: an ego
with more retained connected ties contributes more to the stratum weights.
The estimator above targets
$\text{RR}^{\text{obs}}_{\text{conn}}$. When $\mathcal C$ contains one stratum,
these quantities reduce to the conditional targets of
Sections~\ref{sec:rr-target-overview} and \ref{sec:rr-cde}.

Under the common-parameter convention in Section~\ref{sec:rr-cde},
the armwise argument used to prove Proposition~\ref{prop:conn-bound} gives,
within every retained $c$,
\[
\mu_{1c}^{\text{obs}}
\leq B(R_1^{\text{CDE}},R_{U,1}^{\text{conn}})
\mu_{1c}^{\text{conn}},
\qquad
\mu_{0c}^{\text{conn}}
\leq B(R_0^{\text{CDE}},R_{U,0}^{\text{conn}})
\mu_{0c}^{\text{obs}}.
\]
Averaging these inequalities with the same nonnegative weights $w_c$ in both
arms preserves the original connected-CDE bound:
\[
\text{CRR}^{\text{conn}}_{\text{CDE}}
\geq
\frac{\text{RR}^{\text{obs}}_{\text{conn}}}
{BF_{\text{CDE,conn}}}.
\]
Under Assumption~\ref{assump:forced-contact-transport}, the same argument with
the composed forced-contact factor from
Appendix~\ref{app:fc-compositional} gives
\[
\text{CRR}^{\text{fc}}_{\text{CDE}}
\geq
\frac{\text{RR}^{\text{obs}}_{\text{conn}}}
{BF_{\text{CDE,fc}}}.
\]
Thus, common-weight adjustment is part of the estimator and target; it does
not require stratum-specific sensitivity parameters or a separate bound.

Each retained stratum must contain both exposure arms. A working outcome model
could be used instead, and standardized over the same target distribution, but the formal
inference below concerns empirical estimators.

The latent sensitivity parameters in Section \ref{sec:rr-cde} are not identified
from the observed connected-dyad table when $U$ is unobserved. Thus, they must be specified by the analyst using substantive knowledge.

\subsection{Actor-sharing inference}

The results in this subsection cover realized-frame causal ratios
conditional on the design $\mathcal F_n$ defined next. They use the
finite-frame identifying assumptions stated below; the ego-cluster
bootstrap does not.

Actor-sharing sandwich estimator inference for dyadic data has been developed by \cite{fafchamps_gubert_2007}, \cite{aronow_samii_assenova_2015} and \cite{tabord_meehan_2019}; our fixed-design development below follows the dependency-graph asymptotics of Tabord-Meehan while addressing a different studentization issue induced by heterogeneous pair-specific conditional means.

Inference is stated for a triangular array of networks indexed by the
node count $n$ of Section~\ref{sec:notation-setup}. Actors are labeled
$1,\ldots,n$. Write $A_n$ for the realized $n\times n$ adjacency matrix,
$\mathbf Y^{t-1}$ for the lagged outcome vector, and $\mathbf X$ for the
measured covariates. Let $\mathcal G_n$ be the full set of retained eligible
ordered-dyad rows, whether connected or not. Let $\mathcal D_n$ be the set of
retained unordered actor pairs that contribute at least one directed connected
row, and write $m_n=|\mathcal D_n|$.

Let the complete pre-outcome design be
\[
\mathcal F_n
=
\sigma\!\left(
A_n,\mathbf Y^{t-1},\mathbf X,\mathcal G_n,\mathcal D_n
\right).
\]
These objects are fixed. The remaining randomness is the
second-wave outcomes. Fixed latent attributes may likewise be
treated as part of this sequence, although the analyst
does not observe them.

Let $\mathcal G_{c,n}$ and $\mathcal R_{c,n}$ be the full eligible and
connected ordered-dyad rows in stratum $c$ respectively. Write
$M_{c,n}=|\mathcal G_{c,n}|$, $N_{c,n}=|\mathcal R_{c,n}|$, and
$P_{c,n}=N_{c,n}/\sum_{c'}N_{c',n}$. For $r=(i,j)\in\mathcal G_{c,n}$, write
\[
\rho_{r,y}
=
P\{Y_i^t(y,1;H_{ij})=1\mid\mathcal F_n\}.
\]
\begin{assumption}[Finite-frame causal bridge]
\label{assump:finite-frame-bridge}
Consistency and positivity hold. Within each latent type $u$ and measured
stratum $c$, the average of $\rho_{r,y}$ among dyads naturally observed with
$Y_j^{t-1}=y$ equals its average over all naturally connected target dyads of
that $(u,c)$ type.
\end{assumption}
This is the realized-frame analogue of
Assumptions~\ref{assump:causal-bridge} and
\ref{assump:background-mean-bridge}. It permits arbitrary effect modification
by $H_{ij}$ subject to the same within-$(U,c)$ mean bridge.
It is used only by the actor-sharing results.

\begin{assumption}[Finite-frame forced-contact transport]
\label{assump:finite-frame-fc-transport}
Let $\mathcal R_{uc,n}$ and $\mathcal G_{uc,n}$ be the connected and full
eligible subsets of type $(u,c)$. For $y\in\{0,1\}$,
\[
\frac{1}{|\mathcal R_{uc,n}|}\sum_{r\in\mathcal R_{uc,n}}\rho_{r,y}
=
\frac{1}{|\mathcal G_{uc,n}|}\sum_{r\in\mathcal G_{uc,n}}\rho_{r,y}.
\]
The connected exposure-arm, connected-target, and full-frame latent-type
distributions have common support.
\end{assumption}
This is the realized-frame analogue of
Assumption~\ref{assump:forced-contact-transport}.
It is used only by the forced-contact actor-sharing corollary;
the forced-contact bootstrap uses
Assumption~\ref{assump:forced-contact-transport} instead.

Define the realized-frame causal arm risks and their ratio by
\[
\bar\rho^{\mathrm{conn}}_{yc,n}
:=\frac{1}{N_{c,n}}\sum_{r\in\mathcal R_{c,n}}\rho_{r,y},
\qquad
R^{\mathrm{conn}}_{y,n}
:=\sum_cP_{c,n}\bar\rho^{\mathrm{conn}}_{yc,n},
\qquad
\text{CRR}^{\mathrm{conn}}_n(\mathcal F_n)
:=\frac{R^{\mathrm{conn}}_{1,n}}{R^{\mathrm{conn}}_{0,n}}.
\]
For the forced-contact benchmark, define
\[
\bar\rho^{\mathrm{fc}}_{yc,n}
:=\frac{1}{M_{c,n}}\sum_{r\in\mathcal G_{c,n}}\rho_{r,y},
\qquad
R^{\mathrm{fc}}_{y,n}
:=\sum_cP_{c,n}\bar\rho^{\mathrm{fc}}_{yc,n},
\qquad
\text{CRR}^{\mathrm{fc}}_n(\mathcal F_n)
:=\frac{R^{\mathrm{fc}}_{1,n}}{R^{\mathrm{fc}}_{0,n}}.
\]
Let $\text{RR}^{\mathrm{obs}}_n(\mathcal F_n)$ be the analogous common-weight
ratio of conditional observed cell risks on this frame. Let
$BF_{\mathrm{CDE,conn},n}$ be the single factor generated by sensitivity
restrictions that hold uniformly over the retained strata on this frame.
The connected factor uses the formula in Section~\ref{sec:rr-cde}, and
$BF_{\mathrm{CDE,fc},n}$ uses the composed formula in
Appendix~\ref{app:fc-compositional}, with latent-stratum risks replaced by
type-specific averages of $\rho_{r,y}$ and with $p_y$, $p_A$, and $p$ replaced
by the corresponding latent-type proportions in the realized exposure-arm,
connected, and full eligible frames. The
finite-frame bridge and the same common-weight argument give the deterministic
inequality
\[
\text{CRR}^{\mathrm{conn}}_n(\mathcal F_n)
\geq
\frac{\text{RR}^{\mathrm{obs}}_n(\mathcal F_n)}
{BF_{\mathrm{CDE,conn},n}}.
\]

Under Assumption~\ref{assump:finite-frame-fc-transport}, the corresponding
forced-contact restrictions, also uniform over retained strata, give
\[
\text{CRR}^{\mathrm{fc}}_n(\mathcal F_n)
\geq
\frac{\text{RR}^{\mathrm{obs}}_n(\mathcal F_n)}
{BF_{\text{CDE,fc},n}}.
\]

\begin{assumption}[Conditional actor dissociation]
\label{assump:actor-sharing-sampling}
Let $W_d$ collect the directed-row outcomes of $d\in\mathcal D_n$.
For $\mathcal D,\mathcal D'\subset\mathcal D_n$ with
$d\cap e=\varnothing$ for all $d\in\mathcal D$ and $e\in\mathcal D'$,
\[
(W_d)_{d\in\mathcal D}
\ \indep\
(W_e)_{e\in\mathcal D'}
\ \mid\ \mathcal F_n.
\]
The set $\mathcal C$ is finite and fixed. The positivity, nondegeneracy, and
rate conditions of Appendix~\ref{app:dyadic-inference} hold.
\end{assumption}
Equivalently, sets of dyads that do not share actors are independent given
$\mathcal F_n$. This is the overlapping-cluster restriction of
\cite[Assumption 2.1]{tabord_meehan_2019}.

The causal target and sampling law are separate:
Assumption~\ref{assump:actor-sharing-sampling} is not implied by the
focal-component estimand or causal bridge. It covers reuse of an actor as ego
or alter, including reciprocal rows, but excludes residual dependence
between actor-disjoint pairs. Because no network-formation
model is specified, inference covers repeated second-wave innovations on the
realized design, not draws of another network, baseline vector, or legislative
environment. Conditional means may still depend on the entire fixed network;
only residual innovations are required to dissociate across disjoint actor
sets.

For implementation, bundle the directed rows by their unordered actor pair
$d\in\mathcal D_n$. Let $\mathcal R_d$ contain the directed rows of $d$, and,
for $a=(i,j)\in\mathcal R_d$, write
$O_{da}=Y_i^t$, $Z_{da}=Y_j^{t-1}$, and $C_{da}=c_{ijt}$. Define
\[
Q_{dyc}=\sum_{a\in\mathcal R_d}I(C_{da}=c,Z_{da}=y),\qquad
M_{dyc}=\sum_{a\in\mathcal R_d}I(C_{da}=c,Z_{da}=y)O_{da},
\]
and sum over $d$ to obtain $Q_{yc}$ and $M_{yc}$. With
$w_c=P_{c,n}$, set $\widehat r_{yc}=M_{yc}/Q_{yc}$ and
$\widehat R_y=\sum_cw_c\widehat r_{yc}$ as in the common-weight estimator.
Inference begins from the log of the observed connected-dyad risk ratio,
$\log\widehat{\text{RR}}^{\text{obs}}_{\text{conn}}=\log\widehat R_1-\log\widehat R_0$.
To construct its pair score, first define pair $d$'s centered contribution
to the risk in arm $y$,
\[
\widehat\psi_{y,d}
=
\sum_cw_c
\frac{M_{dyc}-\widehat r_{yc}Q_{dyc}}{Q_{yc}}.
\]
The numerator is the pair's observed outcome total minus the total expected
if the pair had the pooled outcome risk in cell $(y,c)$; division by
$Q_{yc}$ and averaging with $w_c$ place this residual contribution on the
common-weight risk scale. A small change in $R_y$ produces a change divided
by $R_y$ on the log-risk scale, since the derivative of $\log R_y$ is
$1/R_y$. The empirical pair score for the log risk ratio is therefore
\[
\widehat\phi_d
=
\frac{\widehat\psi_{1,d}}{\widehat R_1}
-
\frac{\widehat\psi_{0,d}}{\widehat R_0}.
\]
The formal Taylor expansion, substitution, and remainder justification are
given in \eqref{eq:fixed-design-linearization} in
Appendix~\ref{app:dyadic-inference}.
For studentization, pair scores are next aggregated to actors because,
conditional on $\mathcal F_n$, residual dependence is organized through
shared actors rather than through dyads in isolation.
The actor score is
\[
\widehat\Phi_g=\sum_{d:g\in d}\widehat\phi_d.
\]
The formal positive-semidefinite variance is
\[
\widehat\sigma_{\log\text{RR},\mathrm{AS}}^2
=
\sum_{g=1}^n\widehat\Phi_g^2.
\]
For comparison, the sharper inclusion--exclusion quantity is
\[
\widehat\sigma_{\log\text{RR},\mathrm{IE}}^2
=
\sum_g\widehat\Phi_g^2-\sum_d\widehat\phi_d^2.
\]
It is reported as a secondary analysis but is not automatically valid under
fixed-design mean heterogeneity.

The proof has two nonstandard steps. Conditional on $\mathcal F_n$, centered
unordered-pair scores form a dependency graph in which only pairs sharing an
actor are adjacent. The rate conditions in
Appendix~\ref{app:dyadic-inference} allow Janson's dependency-graph theorem
\cite[Theorem~2]{janson_1988},
together with Theorem~\ref{thm:dependency-graph-clt} and a log-risk-ratio
linearization, to give a conditional central
limit theorem. This is the argument of
\cite[Propositions 3.1--3.2]{tabord_meehan_2019}, applied here to
linearized log-risk-ratio scores rather than ordinary-least-squares
residuals. Studentization likewise follows the dyadic-robust quadratic form
studied by \cite[Theorem 3.1]{tabord_meehan_2019}. The additional
difficulty is that empirical residualization centers the scores only after
summing over pairs, not pair by pair, so heterogeneous pair means produce
an indefinite extra term in the inclusion--exclusion quadratic form and can
make it underestimate variance. The actor-sum quadratic form is instead
positive semidefinite, so its corresponding extra terms are nonnegative and
yield conservative one-sided studentization.
Combining that observed-risk limit with the deterministic sensitivity
inequality above gives the causal lower limit. Appendix~\ref{app:dyadic-inference}
provides the details with complete proofs.

\begin{theorem}[Conservative one-sided actor-sharing inference]
\label{cor:dyadic-lower-limit}
Under Assumptions~\ref{assump:actor-sharing-sampling} and
\ref{assump:cell-total-covariance}, Corollary~\ref{cor:fixed-design-clt}
gives the conditional dependency-graph limit for the fixed-design estimator,
and Lemma~\ref{lem:fixed-design-studentization} gives, for the conditional
variance $\sigma_n^2$ of the centered scaled score,
\[
\frac{\widehat\sigma_{\log\text{RR},\mathrm{AS}}^2}
{\sigma_n^2/m_n^2}
\geq1+o_p(1).
\]
Fix an $\mathcal F_n$-measurable factor
$BF\geq BF_{\text{CDE,conn},n}$, choose
$0<\alpha<1/2$, and impose
Assumption~\ref{assump:finite-frame-bridge} and the finite-frame sensitivity
restrictions defining $BF_{\text{CDE,conn},n}$. Let $z_{1-\alpha}$ be the
$1-\alpha$ standard-normal quantile. Then
\[
\underline L_{1-\alpha}^{\mathrm{AS}}(BF)
=
\frac{
\exp\!\left\{
\log\widehat{\text{RR}}^{\text{obs}}_{\text{conn}}
-z_{1-\alpha}\widehat\sigma_{\log\text{RR},\mathrm{AS}}
\right\}
}{BF}
\]
satisfies
\[
\liminf_{n\to\infty}
P\!\left\{
\text{CRR}^{\mathrm{conn}}_n(\mathcal F_n)
\geq \underline L_{1-\alpha}^{\mathrm{AS}}(BF)
\mid\mathcal F_n
\right\}
\geq 1-\alpha.
\]
\end{theorem}

\begin{corollary}[Forced-contact actor-sharing inference]
\label{cor:dyadic-fc-lower-limit}
Under Assumptions~\ref{assump:actor-sharing-sampling},
\ref{assump:cell-total-covariance},
\ref{assump:finite-frame-bridge}, and
\ref{assump:finite-frame-fc-transport}, fix an $\mathcal F_n$-measurable
$BF\geq BF_{\text{CDE,fc},n}$ satisfying the forced-contact sensitivity
restrictions and choose $0<\alpha<1/2$. Then
\[
\liminf_{n\to\infty}
P\!\left\{
\text{CRR}^{\mathrm{fc}}_n(\mathcal F_n)
\geq \underline L_{1-\alpha}^{\mathrm{AS}}(BF)
\mid\mathcal F_n
\right\}
\geq 1-\alpha.
\]
\end{corollary}

These conditional results cover the connected and forced-contact causal risk
ratios on their realized frames. Coverage for the corresponding population
common-weight targets under repeated draws of networks, baseline outcomes, or
environments would require an additional design-sampling law and corresponding
bridge conditions; none is imposed here.

The actor-sum standard error is conservative rather than generally consistent.
Under condition~\eqref{eq:mean-heterogeneity-condition} in
Appendix~\ref{app:dyadic-inference}, the sharper inclusion--exclusion standard
error is consistent and gives asymptotically exact normal studentization.

For the primary finite-sample implementation, we use the degrees-of-freedom
correction of \cite[Section 4]{tabord_meehan_2019}. Let $M_g$ be the number of
retained unordered pair bundles incident to actor $g$, let
$\mathcal M^H=\max_gM_g$, and define
\[
\kappa
=
n\frac{\mathop{\mathrm{median}}_g M_g}{\mathcal M^H}.
\]
When $\kappa>0$, the reported lower limit replaces $z_{1-\alpha}$ by the
$1-\alpha$ quantile of a $t_\kappa$ distribution. This
finite-sample correction is ad hoc for the risk ratio, not a separate
coverage theorem. Because its critical value exceeds the normal critical
value, it preserves the conservative direction of the formal actor-sum limit.

The actor-sharing theorem is the first of two inferential strategies. The
second, developed next, is an outgoing-ego cluster bootstrap. The two
procedures rest on different sampling assumptions; the bootstrap is not a
consequence of Theorem~\ref{cor:dyadic-lower-limit}.
In particular, the bootstrap does not use
Assumptions~\ref{assump:finite-frame-bridge} or
\ref{assump:finite-frame-fc-transport}: it transfers coverage through
the population bounds of Section~\ref{sec:rr-cde}.

\subsection{Ego-cluster bootstrap}

\begin{assumption}[Independent outgoing-ego clusters]
\label{assump:ego-cluster-bootstrap}
Conditional on a fixed target-eligible ego frame $\mathcal E_n$, with
$n_E=|\mathcal E_n|\to\infty$, assign each directed row uniquely to its outgoing
ego. For ego $i$, collect the nonredundant totals
$(Q_{i0c},Q_{i1c},M_{i0c},M_{i1c})_{c\in\mathcal C}$ into $T_{n,i}$; the stratum
total is recovered as $P_{ic}=Q_{i0c}+Q_{i1c}$. An eligible ego with no retained
row contributes the zero vector. For each $n$, the
$\{T_{n,i}:i\in\mathcal E_n\}$ are independent and identically distributed.
\end{assumption}

\begin{assumption}[Ego-cluster bootstrap regularity]
\label{assump:ego-cluster-regularity}
Under Assumption~\ref{assump:ego-cluster-bootstrap}, there are constants
$a_n>0$ and $\delta>0$ such that the centered vectors
$a_n^{-1}\{T_{n,i}-E(T_{n,i})\}$ have uniformly bounded
$(2+\delta)$ moments,
\[
\mu_n:=a_n^{-1}E(T_{n,i})\to\mu,\qquad
\operatorname{Var}\!\left\{a_n^{-1}T_{n,i}\right\}\to\Sigma,
\]
where $\mu$ lies in the interior of the positive finite-cell parameter space
and $\Sigma$ may be positive semidefinite. Writing $h$ for the smooth map from
mean totals to the log risk ratio, the scalar delta-method variance
$\nabla h(\mu)^\top\Sigma\nabla h(\mu)$ is strictly positive.
\end{assumption}
These are the Lindeberg--Feller and delta-method conditions used in the
proof of Theorem~\ref{cor:ego-bootstrap-lower-limit}.

Assumption~\ref{assump:ego-cluster-bootstrap} is deliberately stronger than
Assumption~\ref{assump:actor-sharing-sampling}. It rules out residual dependence
between different outgoing-ego blocks even when their rows share an alter,
represent reciprocal directions, or place the same actor in different roles.
The eligible frame is resampled in full, including zero-row egos. This preserves
the directed point estimator but does not reproduce the actor-sharing
covariance structure.

Because the cluster law may vary with $n$, the bootstrap targets the
corresponding superpopulation ratio at each $n$:
\[
\ell_n:=h(\mu_n),\qquad
\text{RR}^{\mathrm{obs}}_{\mathrm{conn},n}:=\exp(\ell_n).
\]
This is the ratio of common-weight risks formed from expected cluster totals.
Write $\text{CRR}^{\mathrm{conn}}_{\mathrm{CDE},n}$,
$\text{CRR}^{\mathrm{fc}}_{\mathrm{CDE},n}$, and
$BF_{\mathrm{CDE,conn},n}$, $BF_{\mathrm{CDE,fc},n}$ for the population
targets and sensitivity factors of Section~\ref{sec:rr-cde} under that same
law. These are superpopulation quantities, not the realized-frame targets
of the actor-sharing theorem. Convergence $\mu_n\to\mu$ alone does not
justify centering a root-$n_E$ limit at $h(\mu)$.

\begin{theorem}[One-sided ego-cluster bootstrap]
\label{cor:ego-bootstrap-lower-limit}
Suppose Assumptions~\ref{assump:ego-cluster-bootstrap} and
\ref{assump:ego-cluster-regularity} hold and fix
$0<\alpha<1$. Draw
$n_E$ outgoing-ego blocks with replacement from the empirical distribution and
recompute the finite-stratum estimator in each bootstrap sample. If
$q_\alpha^\star$ is the conditional $\alpha$ quantile of
$\widehat{\text{RR}}_{\text{conn}}^{\text{obs},\star}$, then, for a fixed
nonrandom $BF\geq BF_{\text{CDE,conn},n}$ for all sufficiently large $n$,
with the connected sensitivity restrictions and
Assumptions~\ref{assump:causal-bridge} and
\ref{assump:background-mean-bridge} holding under each such population law,
\[
\underline L^{\text{ego}}_{1-\alpha}(BF)
:=
\frac{q_\alpha^\star}{BF}
\]
satisfies
\[
\liminf_{n\to\infty}
P\!\left\{
\text{CRR}^{\text{conn}}_{\text{CDE},n}
\geq\underline L^{\text{ego}}_{1-\alpha}(BF)
\right\}\geq1-\alpha.
\]
\end{theorem}

\begin{corollary}[Forced-contact ego-cluster bootstrap]
\label{cor:ego-fc-bootstrap-lower-limit}
Under Assumptions~\ref{assump:ego-cluster-bootstrap},
\ref{assump:ego-cluster-regularity},
\ref{assump:causal-bridge}, \ref{assump:background-mean-bridge}, and
\ref{assump:forced-contact-transport}, and under common latent support in
each population law, let $0<\alpha<1$ and let $BF$ be fixed and nonrandom,
with $BF\geq BF_{\text{CDE,fc},n}$ and the forced-contact sensitivity
restrictions holding for all sufficiently large $n$. Then
$\underline L^{\mathrm{ego}}_{1-\alpha}(BF)=q_\alpha^\star/BF$ is an
asymptotic one-sided lower confidence limit for
$\text{CRR}^{\mathrm{fc}}_{\text{CDE},n}$, in the same limit-inferior sense.
\end{corollary}

For a fixed population observed log risk ratio $\ell$, the same argument
applies if $\sqrt{n_E}(\ell_n-\ell)\to0$. This condition holds automatically
when the population observed ratio is constant in $n$. Coverage then transfers
to a fixed population CDE through its corresponding population sensitivity
inequality and a fixed factor dominating that population's bias factor.

When the outgoing-ego independence assumption is acceptable, the bootstrap
can be preferred in small networks, where a normal approximation to the
actor-sharing score may be unreliable. When the network is large and when
residual dependence through a shared alter cannot be ruled out, the
actor-sharing approximation is the more natural default: it uses the weaker
Assumption~\ref{assump:actor-sharing-sampling} and is the paper's primary
formal result.

All confidence limits above use prespecified
target-specific factor $BF$; they do not account for estimating or selecting
$BF$ using the second wave outcomes.

\section{Simulation Study}\label{sec:simulation-study}

We now study the CDE-focused sensitivity analysis
under plausible network generating processes. We consider a setting
where contagion, latent homophily, and outcome susceptibility are varied
separately. This enables us to consider how conservative the Smith-style
selection-bias lower bounds are, in typical settings.

\subsection{Simulation Specification}

We report the main point-estimate simulation at $n=250$ and evaluate one-sided
inference in the same design. The parameters are as
given in Table \ref{tab:simulation-parameters}. Each node has a binary fixed latent covariate
\[
U_i \sim \text{Bernoulli}(p_U).
\]

The edge model is a logistic dyadic model used to draw one fixed symmetric
adjacency matrix. For pair $\{i,j\}$,
\[
\text{logit}\{P(A_{ij}=1)\}
=
\alpha_A
+ \eta_U I(U_i=U_j).
\]
Edges depend only on the latent type, consistent with the DAG
(Figure~\ref{fig:ST_DAG}) in which the outcomes do not influence tie formation.
Conditional on $(U_{i,j},c)$, the focal edge draw is independent of the
nonfocal background used by the outcome model, so the simulation satisfies
Assumption~\ref{assump:forced-contact-transport}.
The intercept $\alpha_A$ is recalibrated in each grid cell to target mean degree
3. The homophily grid is given in Table \ref{tab:simulation-parameters}.
For $\eta_U$, negative values make ties less likely between nodes of the same latent type,
positive values make ties more likely between nodes of the same latent type,
and zero removes latent homophily from the edge model.

Given the fixed network and outcomes at the previous time point, define the lagged neighbor mean
$\bar{Y}_{N_i}^{t-1}
=
d_i^{-1}\sum_{j\neq i} A_{ij}Y_j^{t-1}$
 where $d_i=\sum_{j\neq i}A_{ij}$ is the degree of node $i$, with
$\bar{Y}_{N_i}^{t-1}=0$ when $d_i=0$.

The node outcome model is also logistic:

\[
\text{logit}\{P(Y_i^t=1)\}
=
\alpha_Y
+ \beta_U U_i
+ \beta_P Y_i^{t-1}
+ \beta_N \bar{Y}_{N_i}^{t-1},
\]

Table \ref{tab:simulation-parameters} collects the
fixed parameter values and grid values used in the simulation. For the first wave, $Y_{i}^{t-1}$ and $\bar{Y}_{N_i}^{t-1}$ are taken to be zero, so $Y^{t-1}$ for $t=1$ is generated from $U$ only (no contagion). The $U$s are generated first; then $Y^{t-1}$ is generated from $U$ only; then $A$ is drawn once and held fixed; then the second wave $Y^{t}$ is generated from the outcome model on that fixed network.

\begin{table}[H]
\centering
\caption{Simulation design parameters for the fixed-network simulation study.}
\label{tab:simulation-parameters}
\begin{tabular}{p{0.24\linewidth}p{0.22\linewidth}p{0.43\linewidth}}
\hline
Quantity & Value & Role \\
\hline
\hline
$p_U$ & $0.5$ & Bernoulli probability for the latent binary nodal trait. \\
\hline
$\alpha_A$ & calibrated & Edge-model intercept chosen to give expected mean
degree approximately 3 for each grid cell. \\
$\eta_U$ & $\{-5,-2,-1,0,1,2,5\}$ & Homophily grid; coefficient on
$I(U_i=U_j)$ in the edge model. \\
\hline
$\alpha_Y$ & $-1.1$ & Outcome-model intercept. \\
$\beta_U$ & $\{0,1,2\}$ & Latent susceptibility grid; no, moderate, and strong
$U\to Y$ scenarios. \\
$\beta_P$ & $0.5$ & Coefficient on ego's prior outcome $Y_i^{t-1}$. \\
$\beta_N$ & $\{0,1,2,5\}$ & Contagion grid; coefficient on the lagged neighbor
mean $\bar{Y}_{N_i}^{t-1}$. \\
\hline
\end{tabular}
\end{table}

The homophily problem is activated when the latent trait
also affects the outcome, so we refer to the three $\beta_U$ values as the no,
moderate, and strong $U\to Y$ scenarios.

Each simulated panel is converted into an ego-alter dyad-transition data set.
For the one transition matching the DAG ($Y^{t-1}$ to $Y^{t}$) and each ordered ego-alter pair $(i,j)$ with
$i\neq j$, the row is
\[
D_{ijt}
=
\left(
Y_i^t,\;Y_i^{t-1},\;Y_j^{t-1},\;A_{ij},\;U_i,\;U_j
\right).
\]

Thus, each simulated panel contributes $n(n-1)$ ego-alter dyad-transition rows.
The network is undirected and each row defines $i$ as the ego whose outcome is modeled and $j$ as the alter whose prior outcome may contribute to the ego's neighbor mean.

The ego's own prior outcome $c=Y_i^{t-1}$ is treated as the measured
pre-exposure covariate that enters the adjustment set.

The design grid crosses the $\beta_U$, $\eta_U$, and $\beta_N$ values shown in
Table \ref{tab:simulation-parameters}.
There are therefore $3\times 7\times 4=84$ grid cells. For the point-estimate
heatmaps, each grid cell is run for 1,000 independent Monte Carlo replicates.
Oracle truths and observed contrasts are available in all replicates.
Sensitivity factors additionally require exact empirical support in every
retained stratum. If $\mathcal V_Q$ is the set of support-valid replicates for
quantity $Q$, the plotted cell value is
$\bar Q=|\mathcal V_Q|^{-1}\sum_{m\in\mathcal V_Q}Q^{(m)}$; the software
records $|\mathcal V_Q|$ separately and reports no mean when it is zero.
Support is sparse at the extreme homophily values: across the 12 cells at each
value, the connected factor is valid in an average of 87 replicates at
$\eta_U=-5$ and 100 at $\eta_U=5$, while the forced-contact factor averages
fewer than one and 23 valid replicates, respectively. At
$\eta_U\in\{-1,0,1,2\}$, every connected-factor cell has at least 850 valid
replicates. Thus the extreme-cell bound values are conditional diagnostics,
not full-design averages.

Our focus is on identifying contagion under latent homophily. Thus, we consider the CDE sensitivity analysis under different strengths of contagion and latent homophily, and focus on testing if the target ratio is above one.

Appendix~\ref{app:simulation-oracle-cde} defines the oracle CDE targets used as
simulation benchmarks, the dyad-by-dyad intervention that maps the known DGP to
interventional risks, and the oracle implementations of the pooled and
ego-centric sensitivity bounds evaluated in each grid cell.

\subsection{Simulation Results}

The simulation results summarize how the observed connected-dyad contrast and
the sensitivity-corrected lower bounds behave across the contagion, homophily,
and latent-susceptibility grids.
Figure~\ref{fig:sim-connected-four-panel-moderate} shows the grid-level
point-estimate patterns for the moderate $U\to Y$ scenario. The strong and no
$U\to Y$ scenarios are shown in Appendix~\ref{app:simulation-results} in
Figures~\ref{fig:sim-connected-four-panel-none} and
\ref{fig:sim-connected-four-panel-strong}.

We next discuss the columns of
Figure~\ref{fig:sim-connected-four-panel-moderate}.

\subsubsection{Connected oracle target and observed contrast: Panel Columns 1 and 2}

The oracle naturally connected CDE increases as the contagion coefficient
$\beta_N$ increases. When $\beta_N=0$, the connected CDE is close to the
null across the homophily grid because changing the alter's prior state has no
direct role in the node outcome model. These patterns
are the intended benchmark: a useful procedure should not create evidence of
contagion in the $\beta_N=0$ cells, but should retain power when $\beta_N>0$.

The observed connected-dyad risk-ratio is standardized within levels of the
ego's own prior outcome $c=Y_i^{t-1}$ leaving the part of the
connected-dyad association that survives conditioning on $c=Y_i^{t-1}$. Homophily
changes the composition of exposed and unexposed connected dyads, so the
contrast either exaggerates or attenuates the
CDE depending on the sign of the latent homophily shift.

\subsubsection{Pooled connected CDE lower bound: Panel Column 3}

The pooled connected lower bound divides the observed connected-dyad risk
ratio by the connected-dyad imbalance factor. It should be closest to the raw contrast when latent outcome-risk
variation and latent composition imbalance are weak, and should move downward as
those quantities become stronger. This is the desired behavior of the method:
the lower bound answers whether the observed association is too large to be
explained by the latent homophily.

The practical tradeoff is power. In cells with strong contagion and weak latent
confounding, the pooled bound should still often remain above one. In cells with
weak contagion we often see the pooled bound below one.

\subsubsection{Ego-centric connected lower bounds}

The ego-centric connected version sets the ego shift $R_{E,y}$ to one and only
lets alter-type composition vary within ego type. This refinement is valid only
when the surviving imbalance is alter-carried rather than ego-carried.

In these point estimates the ego-centric refinement improves relative to
the pooled bound when there is contagion. However, because the
simulation DGP does not meet assumption \ref{assump:common-ego-type-bias}, the ego-centric connected bound exceeds one in the null ($\beta_N=0$) cells under positive homophily.

\subsubsection{Results Interpretation}

Overall, in the moderate $U\to Y$ scenario the two targets behave quite
differently. Under the transport condition and among the better-supported
$\eta_U\in\{-2,-1,0,1,2\}$ cells, the pooled forced-contact lower bound at
$\beta_N=5$ ranges from $0.78$ to $1.14$. The extreme forced-contact cells
have too little empirical support for stable interpretation, and the
forced-contact bias adjustment otherwise does not certify contagion.

For the naturally connected target the bias adjustment
yields pooled lower bounds of $0.98$--$1.07$ at $\beta_N=2$ and
$1.14$--$1.28$ at $\beta_N=5$, so the connected target is the more informative
estimand. The
ego-centric refinement is less conservative than the pooled bound. However, again it does not prevent latent homophily from inflating the connected risk ratio. 

As the simulation DGP need not keep ego-type
composition balanced across alter-exposure groups, $R_{E,y}$ need not be
one; setting it to one then under-estimates the connected bias.

\subsubsection{One-sided inference}

Appendix~\ref{app:bootstrap-inference} evaluates one-sided inference on the
same grid using 250 independent truth replicates, 250 evaluation replicates,
and 250 ego-bootstrap draws per evaluation replicate. Each cell holds fixed
the finite mean connected bias factor from its support-valid truth
replicates. The primary test rejects
$\text{CRR}^{\text{conn}}_{\text{CDE}}\le 1$ when the actor-sum lower
limit divided by that factor exceeds one. The ego-cluster percentile is the
stronger-assumption check. Figure~\ref{fig:sim-error-power-connected-moderate}
shows the moderate $U\to Y$ scenario;
Figures~\ref{fig:sim-error-power-connected-none} and
\ref{fig:sim-error-power-connected-strong} report the no and strong
latent-susceptibility scenarios.

The positive-semidefinite loading makes the primary test deliberately
conservative. Across the no, moderate, and strong $U\to Y$ scenarios, the raw
actor-sum Type I rates are $0.011$, $0.045$, and $0.317$. Dividing by the fixed
connected factor reduces them to $0.011$, $0$, and $0$. The corresponding
fixed-factor powers are $0.68$, $0.29$, and $0.02$, compared with $0.79$,
$0.35$, and $0.03$ for the stronger-assumption ego-cluster check. Thus the actor-sum
procedure retains useful power when latent susceptibility is absent or
moderate, but not when both susceptibility and the required adjustment are
strong.

The finite-frame bridge need not hold in every simulated network. Because
the DGP is known, we can compute the exact conditional observed risk ratio,
the realized-frame connected CDE, and the oracle sensitivity factor, and
check whether
$\text{RR}^{\mathrm{obs}}_n(\mathcal F_n)/BF_{\mathrm{CDE,conn},n}
\leq\text{CRR}^{\mathrm{conn}}_n(\mathcal F_n)$.
Among 13,860 evaluations with an available oracle check, 275 (2.0\%) violate
this inequality. Such violations are incompatible with the finite-frame
bridge under the oracle sensitivity restrictions. We exclude these
evaluations from the oracle-factor causal coverage summaries below, which
are therefore conditional on a valid oracle inequality. Passing this check
does not establish every within-stratum bridge equality. The check uses
known outcome probabilities, not realized second-wave outcomes; observed-risk
coverage and fixed-factor Type I error and power retain their original
evaluation sets.

Coverage diagnostics support the theoretical distinction between variance
estimators. Actor-sum coverage of the exact realized-design observed
risk ratio is $0.993$, $0.994$, and $0.992$ across the three scenarios. The
sharper inclusion--exclusion coverages are $0.956$, $0.965$, and $0.969$.
Using the replicate-specific oracle factor, actor-sum coverage of the connected
CDE is $0.996$, $1.000$, and $1.000$ among the 4,131, 4,809, and
4,645 evaluations retained by the oracle check (of 7,000 per scenario),
respectively. The oracle
expected-variance ratios for actor sum are $2.18$, $2.22$, and $2.18$,
compared with $1.11$, $1.21$, and $1.26$ for inclusion--exclusion. These results quantify the price of removing
the untestable mean-heterogeneity condition: actor sum supplies the formal
assumption-lean guarantee, while inclusion--exclusion remains an informative
sharper diagnostic.

Type I error and power use the independently calibrated cell-mean factor and
are therefore comparative operating characteristics. The oracle-factor
coverage summaries evaluate the coverage-transfer step of
Theorem~\ref{cor:dyadic-lower-limit} on frames satisfying the oracle causal
inequality. All 84
cell-mean factors are estimable, but support-valid truth counts range from 12
to 250, so results at $\eta_U=\pm5$ should be read cautiously.

\clearpage
\begin{landscape}
\begin{figure}[p]
\centering
\includegraphics[width=\linewidth,height=6in,keepaspectratio]{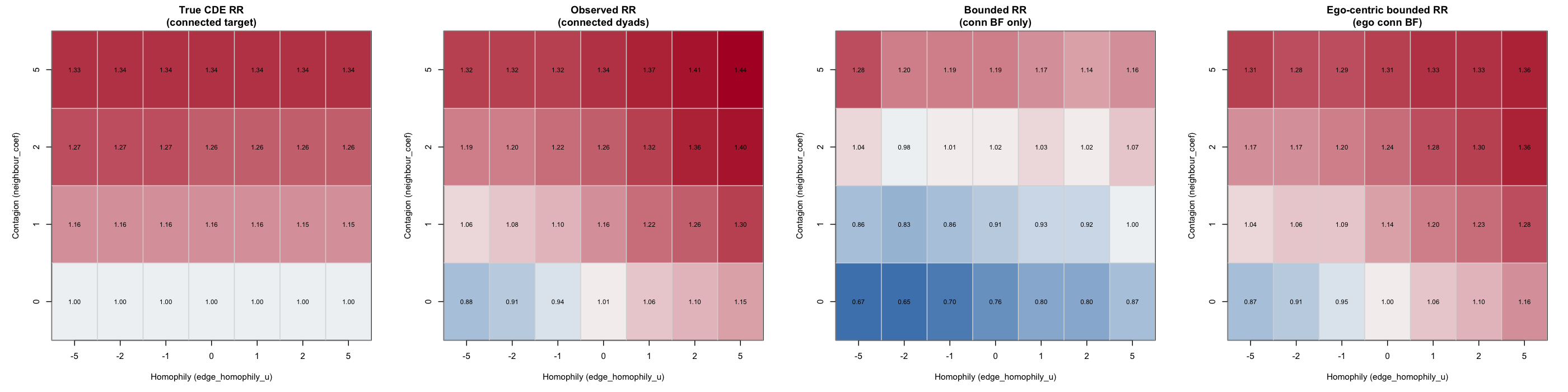}
\caption{Naturally connected CDE simulation heatmaps for the moderate
latent-susceptibility scenario ($n=250$, $\beta_U=1$), with the edge intercept
calibrated to target mean degree 3. Panels show the oracle connected CDE truth,
observed connected-dyad risk ratio, pooled connected lower bound, and
ego-centric connected lower bound. Within each heatmap, the horizontal axis
varies latent homophily in the edge model and the vertical axis varies the
contagion coefficient. All panels use a shared diverging color scale, with
white at RR${}=1$. Truth and observed panels average all 1,000 replicates;
bound panels average only support-valid replicates under the convention in
Section~\ref{sec:simulation-study}.}
\label{fig:sim-connected-four-panel-moderate}
\end{figure}
\end{landscape}
\clearpage

\section{Case Study: TARP Votes under Party-Calibrated Latent Homophily}\label{sec:tarp-application}

Using Office of the Clerk roll calls \cite{congress_xml}, the House rejected
the bailout package on September 29, 2008 (Roll Call 674; 205--228) and passed
a revised Senate-amended package on October 3 (Roll Call 681; 263--171).
Fifty-eight members switched from Nay to Yea. The bill was unpopular,
particularly among Republicans and members in competitive districts
\cite{couch2011financial,montgomery2008house,mian2010political}; the Senate
amendments and four-day interval created a brief period in which colleagues
could exchange information or exert political pressure.

Studies of these votes emphasize campaign contributions, financial-industry
PAC money, and legislators' portfolio exposure
\cite{mian2010political,ramirez2012bailout,tahoun2019wealth}, while
congressional-network studies show that cosponsorship and other ties predict
roll-call similarity \cite{Fowler2006Connecting,Kirkland2011,NBERw16437}.
Neither literature identifies working-tie contagion on this vote. Working
relationships are structured by party, ideology, geography, seniority, and
policy specialization, which also predict bailout votes. A Yea-voting alter
may therefore affect the ego, or the tie may reveal latent characteristics
that make both legislators more likely to vote Yea.

Dynamic-network methods require more outcome and network observations, while
latent-space and stochastic-block approaches require a model for network
formation. We instead ask how strong latent homophily must be to explain the
apparent contagion signal.

The two votes concern slightly different packages four days apart. The CDE is
the effect of a colleague's revealed first-vote position on the ego's vote on
the amended package, holding that four-day environment fixed; it is not the
effect of holding the bill unchanged. Common period effects are absorbed into
the baseline second-vote risk. This is the same separation of environmental
confounding from latent homophily as in \cite{vanderweele2011contagion}.
Differential responsiveness that is a stable trait is charged to latent $U$.
The actor-sharing analysis conditions on this House roster, network, first
vote, party-dyad type, and legislative environment. Residual dependence is
allowed only between dyads that share an actor
(Assumption~\ref{assump:actor-sharing-sampling}). The
separate ego-cluster bootstrap instead invokes an ego-superpopulation
interpretation.

Party is a strong predictor of both working ties and TARP
votes and thus must be included in $c$. We additionally use party as a calibration device for
the latent type $U$: the ego-party outcome-risk variation and connected
party-dyad mix shifts implied by treating party as if it were latent become the
unit of latent confounding. Using the analysis outcomes for this calibration is
not allowed in a formal application of
Theorem~\ref{cor:dyadic-lower-limit}, which requires the sensitivity factor to
be specified without those outcomes. We nevertheless use it to
demonstrate plausible parameter choices because we lack external substantive
knowledge with which to specify the latent-$U$ restrictions. One could instead
calibrate from previous two-step House vote procedures; we do not do that
here. Section~\ref{sec:tarp-substantive-calibration}
reports those restrictions as mild ($\alpha=1/4$), moderate ($\alpha=1/2$),
and strong ($\alpha=1$). We note that an unmeasured confounder as strong as party is very unlikely and suggest the $\alpha=1/4$ and $\alpha=1/2$ cases as reasonable.

Only one legislator that voted Yea on the first vote voted Nay on the second vote. Thus we focus on the switchable egos; those who voted Nay first, $Y_i^1=0$. We therefore target the CDE among this
subpopulation, the scientifically meaningful group for a vote-switching contagion
question. Within that subpopulation the measured covariate is the party-dyad
type, $c=(X_i,X_j)$. We proceed by constructing $A_{ij}$ from filtered cosponsorship,
estimating the party-adjusted $\widehat{\text{RR}}^{\text{obs}}_{\text{conn}}$,
using party's measured strength to calibrate latent $U$. We
focus on the naturally connected target; the corresponding forced-contact
limits are in Appendix~\ref{app:tarp-party-calibration}.

\subsection{Vote outcomes and analytical sample}\label{sec:tarp-votes}

Table \ref{tab:tarp-votes} summarizes the vote roll calls.

\begin{table}[H]
\centering
\caption{TARP House roll calls used in the application.}
\label{tab:tarp-votes}
\begin{tabular}{lllll}
\hline
Occasion & Roll & Date & Bill & Result \\
\hline
First vote  & 674 & 2008-09-29 & H.R.\ 3997 & Failed (205--228) \\
Second vote & 681 & 2008-10-03 & H.R.\ 1424 & Passed (263--171) \\
\hline
\end{tabular}
\end{table}

Let $Y_i^1$ and $Y_i^2$ denote legislator $i$'s vote on the first and second
occasions. The fixed-time panel uses the single transition from occasion $1$ to
occasion $2$, so $Y_i^t$ in the notation of Section~\ref{sec:rr-target-overview}
 corresponds to $Y_i^2$ and $Y_j^{t-1}$ to $Y_j^1$. The sample consists of
$n=433$ members voting on both occasions; $59$ members switched between votes
($58$ Nay$\to$Yea, $1$ Yea$\to$Nay). Vote records and party labels are taken
from the Office of the Clerk XML archive \cite{congress_xml}.

\subsection{Working-relation network}\label{sec:tarp-network}

Working ties $A_{ij}$ are formed from filtered House cosponsorship prior to the
first TARP vote (cutoff 2008-09-29). We retain sponsor--cosponsor events signed
on or before the cutoff, with active status, on non-commemorative HR bills
with fewer than ten cosponsors. Construction
details and bill-filter definitions are given in
Appendix~\ref{app:tarp-network}. Appendix~\ref{app:tarp-network-definitions}
repeats the party-adjusted connected analysis, and the same
party-calibrated latent-homophily factors, under alternative cosponsor caps,
repeated-collaboration and early-cosponsor definitions, and institutional
working-network definitions.

Under this definition the network has mean degree $26.7$, median degree $25$, one isolated
member, and $11{,}556$ connected ordered dyads among TARP voters, of which
$6{,}043$ have a switchable ego ($Y_i^1=0$; see
Section~\ref{sec:tarp-panel}). The degree distribution is shown in
Figure~\ref{fig:tarp-degree}.

\subsection{Dyad panel construction}\label{sec:tarp-panel}

Party is the observed nodal covariate, $X_i\in\{D,R\}$, and enters the
measured adjustment set. Among switchable egos the first vote is constant, so
$c=(X_i,X_j)$ is the party-dyad type.

For each ordered pair $(i,j)$ with $i\neq j$ and $A_{ij}=1$, define one
dyad-transition row:
\[
\big(Y_i^2,\ Y_j^1,\ A_{ij},\ c=(X_i,X_j),\ Y_i^1\big).
\]
We treat $Y_i^1$ and $Y_j^1$ as co-baseline variables with no
within-first-roll-call causal relation; the focal effect is
$Y_j^1\rightarrow Y_i^2$.

Our analysis is restricted to switchable egos, $Y_i^1=0$. The party-adjusted observed risk ratio
standardizes the connected contrast to the connected distribution of
party-dyad types in that subpopulation.

\subsection{Empirical CDE inputs and sensitivity analysis}\label{sec:tarp-empirical}

Following Section~\ref{sec:cde-estimation}, the primary empirical estimand is
the party-adjusted observed connected-dyad risk ratio among switchable egos, denoted $\widehat{\text{RR}}^{\text{obs}}_{\text{conn}}$.

With party-dyad type in $c$ the standardized risks are $\hat{r}_1^{\text{obs}}=0.273$ and $\hat{r}_0^{\text{obs}}=0.237$, thus $\widehat{\text{RR}}^{\text{obs}}_{\text{conn}}=1.15$. The actor-sum one-sided $95\%$ lower limit on that party-adjusted contrast is $1.039$, the
inclusion--exclusion diagnostic is $1.073$, and the ego-cluster percentile
limit is $1.080$. Among switchable egos, a working tie to a Yea-voting alter is
associated with about a $16\%$ higher second-vote Yea risk after party adjustment.
Both formal one-sided limits exclude one, so the party-adjusted connected
contrast identifies positive contagion when $BF=1$.
Table~\ref{tab:tarp-stratum} reports party-dyad risks in the switchable
sample, the inputs for the party-adjusted contrast. Among egos who voted
Yea first, second-vote Yea risk is $0.98$--$1.00$ in every party-dyad cell,
so that stratum is omitted.

\begin{table}[H]
\centering
\caption{Connected-dyad second-vote Yea risks among switchable egos,
by party dyad type, primary cosponsorship network.
}
\label{tab:tarp-stratum}
\begin{tabular}{lcc}
\hline
Party dyad & $\widehat r_1^{\text{obs}}$ & $\widehat r_0^{\text{obs}}$ \\
\hline
$D\!-\!D$ & 0.343 & 0.321 \\
$D\!-\!R$ & 0.306 & 0.264 \\
$R\!-\!D$ & 0.237 & 0.158 \\
$R\!-\!R$ & 0.201 & 0.169 \\
Unadjusted & 0.288 & 0.226 \\
\hline
\end{tabular}
\end{table}

\subsection{Party as a calibration device for latent $U$.}
For our demonstration purposes, we use party to calibrate sensitivity parameters for latent $U$ that would ordinarily have to be specified from substantive knowledge: the
outcome-risk variation across ego party, the imbalance of party-dyad types
among exposed and unexposed connected dyads, and, for the forced-contact
target, the difference between connected dyads and the full ordered-dyad
population.

Ego-party second-vote Yea risks among switchable connected dyads give
$\widehat R_1^{\text{party}}=1.54$ and
$\widehat R_0^{\text{party}}=1.82$: even among switchable egos, the
second-vote Yea risk varies by about $1.5$--$1.8$ across ego party after both
previous votes are held fixed. The connected-dyad party distribution shifts
use the full party-dyad type:
$\widehat R_{1}^{\text{conn,party}}=1.30$ and
$\widehat R_{0}^{\text{conn,party}}=1.31$, yielding
$\widehat{BF}_{\text{conn}}^{\text{party}}=1.219$.
Selection-into-ties shifts relative to all ordered dyads are
$\widehat R_{1}^{\text{sel,party}}=1.75$ and
$\widehat R_{0}^{\text{sel,party}}=1.98$, so
$\widehat{BF}_{\text{fc}}^{\text{party}}=1.724$, using the composition in
Appendix~\ref{app:fc-compositional}. These party-implied ratios are the unit
of latent $U$ confounding in the next subsection.
Appendix~\ref{app:tarp-party-calibration} reports the cell counts and the
max/min calculations.

\subsection{Sensitivity to latent $U$ at a fraction of party strength}\label{sec:tarp-substantive-calibration}

Write $R^{(\alpha)}=1+\alpha(R^{\text{party}}-1)$ for each party-implied
ratio, and let $BF(\alpha)$ be the corresponding product of Ding--VanderWeele
factors. We call $\alpha=1/4$ mild latent homophily, $\alpha=1/2$
moderate, and $\alpha=1$ strong.  With party in $c$, leftover traits such as
ideology, sub-party blocs, or constituency pressure are unlikely to match
that strength, so the strong cell is very conservative.
The primary connected bound is
\[
\widehat{\text{LB}}_{\text{conn}}(\alpha)
=
\frac{\widehat{\text{RR}}^{\text{obs}}_{\text{conn}}}
{BF_{\text{CDE,conn}}(\alpha)}.
\]

We focus on the naturally connected target.
Table~\ref{tab:tarp-formal-limits} reports the party-adjusted contrast and the
three restrictions on latent $U$.
At $\alpha=0$ the party-adjusted contrast is $1.155$ and the actor-sum
endpoint is $1.039$. Under the mild restriction the connected factor is
$BF=1.021$, the point bound is $1.131$, and the actor-sum endpoint is $1.018$.
Under the moderate restriction
($BF=1.071$) the point bound is $1.079$ but the actor-sum endpoint is $0.971$;
under the strong restriction ($BF=1.219$) the point bound is $0.947$ and the
endpoint is $0.852$. Moderate or strong latent homophily is enough to explain
the party-adjusted contrast away.

The forced-contact target is not established at any of these
strengths (Table~\ref{tab:tarp-formal-limits-fc}). Even the mild
forced-contact actor-sum limit is $0.968$.
The ego-centric refinement of Section~\ref{sec:ego-centric-bounds} is not a
primary analysis here: with party-dyad type already in $c$, ego-party
composition is absorbed by the adjustment. The stronger transport condition required for the forced-contact target is also not plausible in this application.

\begin{table}[H]
\centering
\caption{Party-adjusted one-sided $95\%$ lower limits for the naturally connected
target among switchable egos, primary cosponsorship network. Mild, moderate,
and strong restrictions on latent $U$ scale party-implied values by $\alpha=1/4$, $1/2$, and $1$. The positive-semidefinite actor-sum (AS)
limit is the formal actor-sharing implementation. Inclusion--exclusion (IE) is
the sharper secondary limit under condition~\eqref{eq:mean-heterogeneity-condition};
the ego-cluster percentile is the stronger-assumption nonparametric check.}
\label{tab:tarp-formal-limits}
\begin{tabular}{lrrrrr}
\hline
Restriction & $BF$ & Point bound & AS LCL & IE LCL & Ego LCL \\
\hline
Party in $c$, no latent $U$ & $1$ & $1.155$ & $1.039$ & $1.073$ & $1.080$ \\
Mild ($\alpha=1/4$) & $1.021$ & $1.131$ & $1.018$ & $1.051$ & $1.058$ \\
Moderate ($\alpha=1/2$) & $1.071$ & $1.079$ & $0.971$ & $1.002$ & $1.009$ \\
Strong ($\alpha=1$) & $1.219$ & $0.947$ & $0.852$ & $0.880$ & $0.886$ \\
\hline
\end{tabular}
\end{table}

\section{Discussion}

Our method reframes identification of contagion under latent homophily as a
sensitivity-analysis problem for a naturally connected causal estimand. The
bounds make explicit how strongly an unobserved variable must predict the
outcome and sort dyads through network ties to explain away an observed
contrast.

The naturally connected CDE targets contagion among dyads whose ties formed,
holding the focal tie present. This is often the more defensible social-network
estimand because ties are sparse and meaningful. The forced-contact CDE extends
the intervention to all eligible dyads and therefore additionally requires
transport and adjustment for selection from all possible dyads into observed
ties. Its comparison with the connected target makes the cost of that
extrapolation visible.

The simulations illustrate the resulting robustness--power trade-off. Stronger
latent homophily and outcome susceptibility can make the raw connected contrast
overstate contagion; the bounds reduce false-positive conclusions at the cost
of power. The TARP analysis similarly finds that the party-adjusted connected
conclusion survives a mild restriction on latent $U$, but not moderate or
strong restrictions, and does not establish the forced-contact target.

In practice, analysts should report the observed connected contrast with
target-specific sensitivity bounds and state which inferential framework is
credible. Actor-sharing inference conditions on the realized design and permits
dependence through shared actors, whereas the outgoing-ego bootstrap invokes a
stronger superpopulation independence assumption.

Building on \cite{vanderweele2011contagion} and complementing the sequence-level
bounds of \cite{vansteeg_2013}, the framework provides a nonparametric dyadic
selection-bias analysis and one-sided inference for the tractable naturally connected
target. Its scope remains a single fixed network, two observed outcome waves, and binary variables;
extensions to richer longitudinal settings are left for future work.

\section*{Acknowledgements}

The author thanks Medha Uppala for helpful comments during the early stages of this project, Mark Handcock for comments on an early manuscript draft, and an anonymous reviewer of an earlier version for directing attention to \cite{tabord_meehan_2019} and asymptotic dyadic inference.

The author acknowledges the use of AI tools during
the preparation of this manuscript. These tools were used to support editing,
code development, simulation workflow management, and checks of mathematical and
computational exposition. All scientific claims, mathematical arguments,
simulation designs, code, and final text were reviewed and are the responsibility
of the author.

\newpage

\appendix
\section{Forced-contact compositional sensitivity analysis}\label{app:fc-compositional}

This appendix gives the one-step direct factor used in
Proposition~\ref{prop:direct-fc-bound} and then develops a compositional factor
whose sensitivity parameters can be calibrated separately.

\subsection{Direct one-step factor}\label{app:fc-direct}

Under Assumption~\ref{assump:forced-contact-transport}, the forced-contact
target has the latent-standardized form
\[
\text{CRR}^{\text{fc}}_{\text{CDE}}
:=
\frac{
\sum_u r_1^{\text{fc}}(u)p(u)
}{
\sum_u r_0^{\text{fc}}(u)p(u)
}.
\]
This target is stronger because it standardizes to all dyads, not only dyads
that naturally formed ties. Its bias relative to the observed connected-dyad
contrast is
\begin{align*}
\frac{\text{RR}^{\text{obs}}_{\text{conn}}}
{\text{CRR}^{\text{fc}}_{\text{CDE}}}
&=
\left[
\frac{\sum_u r_1^{\text{fc}}(u)p_1(u)}
{\sum_u r_1^{\text{fc}}(u)p(u)}
\right]
\left[
\frac{\sum_u r_0^{\text{fc}}(u)p(u)}
{\sum_u r_0^{\text{fc}}(u)p_0(u)}
\right].
\end{align*}

Here the bracketed terms explicitly include the difference between the natural
connected distributions $p_1(u),p_0(u)$ and the interventional target
distribution $p(u)$. This captures selection into naturally observed ties:
under homophily, dyads with $A_{ij}=1$ can have a very different latent
composition from the full dyad population under a forced-contact intervention.

One can bound this forced-contact contrast in one step by comparing the exposed
and unexposed connected distributions directly to the full dyad population.
Define the direct forced-contact latent distribution shifts
\[
R_{U,1}^{\text{fc,dir}}
:=
\max_u \frac{p_1(u)}{p(u)},
\qquad
R_{U,0}^{\text{fc,dir}}
:=
\max_u \frac{p(u)}{p_0(u)}.
\]
Using the same outcome-risk variation parameters, define the one-step
diagnostic factor
\[
BF_{\text{CDE,fc}}^{\text{dir}}
:=
B(R_1^{\text{CDE}},R_{U,1}^{\text{fc,dir}})
B(R_0^{\text{CDE}},R_{U,0}^{\text{fc,dir}}).
\]

Substantive knowledge ultimately powers the bounds, so sensitivity parameters
should separate quantities that can be informed by different domain
knowledge. The direct factor is the least conservative of the proposed
forced-contact adjustments, but it bundles within-connected imbalance and
selection into naturally observed ties. We therefore decompose the
forced-contact discrepancy into connected and selection components. The
resulting bound is more conservative but permits the two components to be
calibrated separately.

\subsection{Bias decomposition}\label{app:fc-decomposition}

Under the connected bridge and forced-contact transport condition,
\begin{align*}
\frac{\text{RR}^{\text{obs}}_{\text{conn}}}
{\text{CRR}^{\text{fc}}_{\text{CDE}}}
&=
\left[
\frac{\text{RR}^{\text{obs}}_{\text{conn}}}
{\text{CRR}^{\text{conn}}_{\text{CDE}}}
\right]
\left[
\frac{\text{CRR}^{\text{conn}}_{\text{CDE}}}
{\text{CRR}^{\text{fc}}_{\text{CDE}}}
\right].
\end{align*}
Thus, the total bias to the forced-contact target is the product of the
within-connected bias and the selection bias.

To show the decomposition explicitly, under
Assumption~\ref{assump:forced-contact-transport} write the common within-type
risk as $r_y(u)=r_y^{\mathrm{conn}}(u)=r_y^{\mathrm{fc}}(u)$ and set
$M_y(q)=\sum_u r_y(u)q(u)$. The one-step forced-contact bias is
\begin{align*}
\frac{\text{RR}^{\text{obs}}_{\text{conn}}}
{\text{CRR}^{\text{fc}}_{\text{CDE}}}
&=
\frac{M_1(p_1)/M_0(p_0)}{M_1(p)/M_0(p)}
=
\left[\frac{M_1(p_1)}{M_1(p)}\right]
\left[\frac{M_0(p)}{M_0(p_0)}\right].
\end{align*}
Insert the naturally connected averages $M_1(p_A)$ and $M_0(p_A)$:
\begin{align*}
\frac{M_1(p_1)}{M_1(p)}
&=
\left[\frac{M_1(p_1)}{M_1(p_A)}\right]
\left[\frac{M_1(p_A)}{M_1(p)}\right],
\\
\frac{M_0(p)}{M_0(p_0)}
&=
\left[\frac{M_0(p_A)}{M_0(p_0)}\right]
\left[\frac{M_0(p)}{M_0(p_A)}\right].
\end{align*}
The first bracket in each line is the within-connected comparison involving
$p_1$, $p_0$, and $p_A$. The second bracket in each line is the selection
comparison between the naturally connected distribution $p_A$ and the full dyad
distribution $p$. Multiplying the four brackets recovers the one-step
forced-contact bias, so the bias decomposes exactly through $p_A$ and the
within-connected and selection sensitivity parameters can be specified
separately.

\subsection{Composed sensitivity factor and lower bound}

Define the selection-into-ties distribution shifts
\[
R_{U,1}^{\text{sel}}
:=
\max_u \frac{p_A(u)}{p(u)},
\qquad
R_{U,0}^{\text{sel}}
:=
\max_u \frac{p(u)}{p_A(u)}.
\]
Applying Lemma~\ref{lem:bounding-factor} to each weighted average in the
selection step gives
\[
BF_{\text{CDE,sel}}
:=
B(R_1^{\text{CDE}},R_{U,1}^{\text{sel}})
B(R_0^{\text{CDE}},R_{U,0}^{\text{sel}}).
\]
For the forced-contact target, compose the distribution-ratio parameters
before applying the bounding function:
\[
R_{U,1}^{\text{fc}}
:=
R_{U,1}^{\text{conn}}R_{U,1}^{\text{sel}},
\qquad
R_{U,0}^{\text{fc}}
:=
R_{U,0}^{\text{conn}}R_{U,0}^{\text{sel}},
\]
and define
\[
BF_{\text{CDE,fc}}
:=
B(R_1^{\text{CDE}},R_{U,1}^{\text{fc}})
B(R_0^{\text{CDE}},R_{U,0}^{\text{fc}}).
\]

\begin{proposition}[Composed-ratio forced-contact CDE lower bound]\label{prop:composed-fc-bound}
Under Assumptions~\ref{assump:causal-bridge},
\ref{assump:background-mean-bridge}, and
\ref{assump:forced-contact-transport}, and common latent support of
$p_1$, $p_0$, $p_A$, and $p$, the composed-ratio factor
$BF_{\text{CDE,fc}}$ gives a valid lower bound,
\[
\frac{\text{RR}^{\text{obs}}_{\text{conn}}}{BF_{\text{CDE,fc}}}
\leq
\text{CRR}^{\text{fc}}_{\text{CDE}},
\]
and it is sandwiched between the one-step direct factor and the product of the
two stage-wise factors,
\[
BF_{\text{CDE,fc}}^{\text{dir}}
\leq
BF_{\text{CDE,fc}}
\leq
BF_{\text{CDE,conn}}BF_{\text{CDE,sel}}.
\]
\end{proposition}

The compositional analysis first asks how much latent imbalance remains among
naturally connected dyads, then how different those dyads are from the full
population that would be placed in contact under the intervention. For the
naturally connected target, report
$\text{RR}^{\text{obs}}_{\text{conn}}/BF_{\text{CDE,conn}}$; for the
forced-contact target, report
$\text{RR}^{\text{obs}}_{\text{conn}}/BF_{\text{CDE,fc}}$.
The one-step direct factor is less conservative but generally harder to
calibrate because it combines both sources of imbalance.

\section{Proofs}\label{app:proofs}

Every bound in Section~\ref{sec:rr-cde} and
Appendix~\ref{app:fc-compositional} is an application of the bounding factor of
\cite{ding2016sensitivity}, specialized to the selection problem as in
\cite{smith2019bounding}. We therefore prove the building-block inequality
(Lemma~\ref{lem:bounding-factor}) and then obtain each proposition by applying
it to the corresponding ratio.

\subsection{Ego-centric construction}\label{app:ego-construction}

We write $E=U_i$ for the ego latent type and $V=U_j$ for the alter latent type,
so that the dyad type factors as $U=(E,V)$. Writing $p_y(e,a)=p_y(e)\,p_y(a\mid e)$,
and likewise for the naturally connected target distribution $p_A$,
\[
\frac{p_1(u)}{p_A(u)}
=
\underbrace{\frac{p_1(e)}{p_A(e)}}_{\text{ego shift}}
\;
\underbrace{\frac{p_1(a\mid e)}{p_A(a\mid e)}}_{\text{alter shift}} ,
\]
so that, taking maxima,
\[
R_{U,1}^{\text{conn}}
\le
R_{E,1}\,R_{V,1},
\qquad
R_{E,1}:=\max_e\frac{p_1(e)}{p_A(e)},
\quad
R_{V,1}:=\max_{e,a}\frac{p_1(a\mid e)}{p_A(a\mid e)},
\]
and analogously $R_{U,0}^{\text{conn}}\le R_{E,0}\,R_{V,0}$ with
$R_{E,0}:=\max_e p_A(e)/p_0(e)$ and
$R_{V,0}:=\max_{e,a} p_A(a\mid e)/p_0(a\mid e)$.

With the ego shift set to one (Assumption~\ref{assump:common-ego-type-bias}),
only the alter comparison within ego type remains, and therefore
$R_{U,1}^{\text{conn}}\le R_{V,1}$ and
$R_{U,0}^{\text{conn}}\le R_{V,0}$.
For each ego type $e$,
write $R_y^{\text{CDE}}(e)$ for the outcome-risk variation of
$r_y^{\text{conn}}$ over alter types given $E=e$, and write
$R_{V,1}^{\text{conn}}(e)=\max_a p_1(a\mid e)/p_A(a\mid e)$ and
$R_{V,0}^{\text{conn}}(e)=\max_a p_A(a\mid e)/p_0(a\mid e)$ for the
within-ego versions of the alter shifts above.
The ego-centric factor is the worst-case within-ego bounding-factor product
\[
BF_{\text{CDE,conn}}^{\text{ego}}
:=
\left[
\max_e B\{R_1^{\text{CDE}}(e),R_{V,1}^{\text{conn}}(e)\}
\right]
\left[
\max_e B\{R_0^{\text{CDE}}(e),R_{V,0}^{\text{conn}}(e)\}
\right].
\]

\subsection{Proofs of the bounds}

\begin{proof}[Proof of Lemma~\ref{lem:bounding-factor}]
This is the bounding factor of \cite{ding2016sensitivity}, used in the
selection-bias bounds of \cite[Results~1A and~5A]{smith2019bounding}.
For completeness we record the
argument. Let $M(s)=\sum_u r(u)s(u)$, $m=\min_u r(u)$, and
$\Gamma=R_U=\max_u p(u)/q(u)$. Common support gives $q(u)>0$ on $\mathcal U$,
so each ratio $p(u)/q(u)$ is well-defined. By definition of the maximum,
\[
\frac{p(u)}{q(u)}\le\Gamma
\qquad\text{for every }u\in\mathcal U,
\]
and therefore $p(u)\le\Gamma q(u)$. If $\Gamma=1$, then $p=q$ and
$M(p)/M(q)=1=B(R_Y,1)$. Now take $\Gamma>1$, and define
\[
h(u)=q(u)-\frac{p(u)}{\Gamma}.
\]
The inequality $p(u)\le\Gamma q(u)$ is exactly $h(u)\ge 0$. Summing and using
that $p$ and $q$ are probability distributions,
\[
\sum_u h(u)
=
\sum_u q(u)-\frac{1}{\Gamma}\sum_u p(u)
=
1-\frac{1}{\Gamma}.
\]
Substitute $q(u)=p(u)/\Gamma+h(u)$ into $M(q)$:
\[
M(q)
=
\sum_u r(u)\frac{p(u)}{\Gamma}+\sum_u r(u)h(u)
=
\frac{M(p)}{\Gamma}+\sum_u r(u)h(u).
\]
Because $r(u)\ge m$ and $h$ is nonnegative,
\[
\sum_u r(u)h(u)
\ge
m\sum_u h(u)
=
m\left(1-\frac{1}{\Gamma}\right),
\]
and therefore
\[
M(q)
\ge
\frac{M(p)}{\Gamma}+m\left(1-\frac{1}{\Gamma}\right).
\]
Write $x=M(p)/m$. Then $m\le M(p)\le m R_Y$, so $1\le x\le R_Y$. Both sides of
the displayed lower bound are positive, and rearranging gives
\[
\frac{M(p)}{M(q)}
\le
\frac{M(p)}{M(p)/\Gamma+m(1-1/\Gamma)}
=
\frac{\Gamma x}{x+\Gamma-1}.
\]
Rewrite the right-hand side as
\[
\frac{\Gamma x}{x+\Gamma-1}
=
\Gamma-\frac{\Gamma(\Gamma-1)}{x+\Gamma-1}.
\]
The subtracted term is decreasing in $x$ for $\Gamma>1$, so the right-hand
side is increasing in $x$. Combined with $x\le R_Y$,
\[
\frac{M(p)}{M(q)}
\le
\frac{\Gamma x}{x+\Gamma-1}
\le
\frac{\Gamma R_Y}{R_Y+\Gamma-1}
=B(R_Y,R_U).
\]

For sharpness, take two risk levels $m$ and $m R_Y$, let $p$ concentrate
arbitrarily closely on the high-risk level, and let the residual measure $h$
concentrate on the low-risk level. Arbitrarily small positive masses preserve
common support, so the displayed constant is approached in the closure.
\end{proof}

\begin{proof}[Proof of Proposition~\ref{prop:conn-bound}]
With $M_y^{\mathrm{conn}}(q)=\sum_u r_y^{\mathrm{conn}}(u)q(u)$,
Assumptions~\ref{assump:causal-bridge} and
\ref{assump:background-mean-bridge} give
\[
\frac{\text{RR}^{\text{obs}}_{\text{conn}}}{\text{CRR}^{\text{conn}}_{\text{CDE}}}
=
\left[\frac{M_1^{\mathrm{conn}}(p_1)}{M_1^{\mathrm{conn}}(p_A)}\right]
\left[\frac{M_0^{\mathrm{conn}}(p_A)}{M_0^{\mathrm{conn}}(p_0)}\right].
\]
Applying Lemma~\ref{lem:bounding-factor} to the first bracket with
$r=r_1^{\mathrm{conn}}$,
$p=p_1$, $q=p_A$ gives a bound of $B(R_1^{\text{CDE}},R_{U,1}^{\text{conn}})$, and
to the second bracket with $r=r_0^{\mathrm{conn}}$, $p=p_A$, $q=p_0$ a bound of
$B(R_0^{\text{CDE}},R_{U,0}^{\text{conn}})$. Their product is $BF_{\text{CDE,conn}}$,
so the forward bias is at most $BF_{\text{CDE,conn}}$. Hence
$\text{RR}^{\text{obs}}_{\text{conn}}/BF_{\text{CDE,conn}}\le
\text{CRR}^{\text{conn}}_{\text{CDE}}$.
\end{proof}

\begin{proof}[Proof of Proposition~\ref{prop:direct-fc-bound}]
Assumption~\ref{assump:forced-contact-transport} equates
$r_y^{\mathrm{conn}}$ and $r_y^{\mathrm{fc}}$. The forward bias factorization
is therefore identical to that in the proof of
Proposition~\ref{prop:conn-bound}, with the full dyad distribution $p$ in place
of $p_A$. The two brackets are bounded by
$B(R_1^{\text{CDE}},R_{U,1}^{\text{fc,dir}})$ and
$B(R_0^{\text{CDE}},R_{U,0}^{\text{fc,dir}})$, whose product is
$BF_{\text{CDE,fc}}^{\text{dir}}$. Rearranging gives the stated lower bound.
\end{proof}

\begin{proof}[Proof of Proposition~\ref{prop:composed-fc-bound}]
We have: 
\[
R_{U,1}^{\text{fc,dir}}
=\max_u\frac{p_1(u)}{p(u)}
\le
\Big(\max_u\frac{p_1(u)}{p_A(u)}\Big)\Big(\max_u\frac{p_A(u)}{p(u)}\Big)
=R_{U,1}^{\text{conn}}R_{U,1}^{\text{sel}}
=R_{U,1}^{\text{fc}},
\]
and likewise $R_{U,0}^{\text{fc,dir}}\le R_{U,0}^{\text{fc}}$. Since $B(a,\cdot)$
is increasing, $BF_{\text{CDE,fc}}^{\text{dir}}\le BF_{\text{CDE,fc}}$.
Because $BF_{\text{CDE,fc}}^{\text{dir}}$ bounds the forward bias by
Proposition~\ref{prop:direct-fc-bound}, the larger $BF_{\text{CDE,fc}}$ also
gives a valid lower bound.

For the upper ordering, fix $y$ and write $a=R_y^{\text{CDE}}$,
$b=R_{U,y}^{\text{conn}}$, $c=R_{U,y}^{\text{sel}}$, all at least $1$. A direct
computation gives the identity
\[
a(a+bc-1)-(a+b-1)(a+c-1)=(a-1)(b-1)(c-1)\ge 0,
\]
which rearranges to $B(a,bc)\le B(a,b)\,B(a,c)$. Taking the product over
$y\in\{0,1\}$ yields $BF_{\text{CDE,fc}}\le BF_{\text{CDE,conn}}BF_{\text{CDE,sel}}$.
\end{proof}

\begin{proof}[Proof of Proposition~\ref{prop:ego-bounds}]
Decompose each connected distribution shift into its ego and alter components,
$p_y(e,a)=p_y(e)\,p_y(a\mid e)$, as in
Appendix~\ref{app:ego-construction}. Under
Assumption~\ref{assump:common-ego-type-bias} the ego marginals are common across
the connected exposed, connected unexposed, and naturally connected populations,
so the ego shifts satisfy $R_{E,1}=R_{E,0}=1$ and the bias is a mixture over ego
types with common weights. A ratio of two such common-weight mixtures is bounded
by the largest within-ego-type ratio, so it suffices to bound the bias separately
within each ego type $e$ and then maximize over $e$. Within each $e$,
Lemma~\ref{lem:bounding-factor} applied to the two standardization ratios bounds
the within-ego bias by
$B\{R_1^{\text{CDE}}(e),R_{V,1}^{\text{conn}}(e)\}
B\{R_0^{\text{CDE}}(e),R_{V,0}^{\text{conn}}(e)\}$ for the connected target, and
taking the maximum
over $e$ of each factor yields a bound uniform in the ego type, giving
$BF_{\text{CDE,conn}}^{\text{ego}}$ and
hence the stated lower bound on $\text{CRR}^{\text{conn}}_{\text{CDE}}$.
\end{proof}

\subsection{Actor-sharing inference}\label{app:dyadic-inference}

This subsection supplies the technical details for
Theorem~\ref{cor:dyadic-lower-limit} and
Corollary~\ref{cor:dyadic-fc-lower-limit}. It constructs the centered pair
scores, states the dependency-graph rate conditions, derives the log-risk-ratio
linearization, and formalizes variance expansion and studentization in two
lemmas. The central-limit and inclusion--exclusion studentization arguments
follow \cite{tabord_meehan_2019}; the actor-sum branch of the studentization
lemma records the conservative modification used in
Theorem~\ref{cor:dyadic-lower-limit}. The subsection then exhibits the
mean-heterogeneity obstruction to inclusion--exclusion studentization and
proves the two actor-sharing results.

We condition throughout on $\mathcal F_n$. Let $d=\{g,h\}$ index a retained
unordered pair and $a\in\mathcal R_d$ its directed rows. Write
\[
O_{da}=Y_i^t,\qquad Z_{da}=Y_j^{t-1},\qquad C_{da}=c_{ijt}.
\]
For $y\in\{0,1\}$ and $c\in\mathcal C$, define the fixed exposure count and
random outcome total
\[
Q_{dyc}=\sum_{a\in\mathcal R_d}I(C_{da}=c,Z_{da}=y),\qquad
M_{dyc}=\sum_{a\in\mathcal R_d}I(C_{da}=c,Z_{da}=y)O_{da}.
\]
Let $Q_{yc}=\sum_dQ_{dyc}$, $M_{yc}=\sum_dM_{dyc}$, and let the common
connected-frame weights $w_c$ be fixed. Then
\[
\widehat r_{yc}=\frac{M_{yc}}{Q_{yc}},\qquad
\widehat R_y=\sum_cw_c\widehat r_{yc},\qquad
\widehat\ell_n=\log\widehat R_1-\log\widehat R_0.
\]
Set
\[
m_{dyc}=E(M_{dyc}\mid\mathcal F_n),\quad
m_{yc}=\sum_dm_{dyc},\quad
r_{yc}=\frac{m_{yc}}{Q_{yc}},\quad
R_y=\sum_cw_cr_{yc},
\]
and let
\[
\ell_n=\log R_1-\log R_0
=\log\text{RR}^{\mathrm{obs}}_n(\mathcal F_n).
\]
This target is the ratio of conditional mean risks under repeated second-wave
innovations on the realized design.

Let $m_n=|\mathcal D_n|$ and define the arm errors
\[
\delta_{y,n}:=\widehat R_y-R_y.
\]
On the event $\widehat R_0>0$ and $\widehat R_1>0$, Taylor's theorem gives,
for each $y\in\{0,1\}$,
\[
\log\widehat R_y-\log R_y
=
\frac{\delta_{y,n}}{R_y}
-\frac{\delta_{y,n}^2}{2\widetilde R_y^2},
\]
where $\widetilde R_y$ lies between $\widehat R_y$ and $R_y$. Subtracting
the expansion for the unexposed arm from that for the exposed arm yields
\[
\widehat\ell_n-\ell_n
=
\frac{\delta_{1,n}}{R_1}
-\frac{\delta_{0,n}}{R_0}
+\mathcal R_n,
\qquad
\mathcal R_n
=
-\frac{\delta_{1,n}^2}{2\widetilde R_1^2}
+\frac{\delta_{0,n}^2}{2\widetilde R_0^2}.
\]
Because the counts and weights are fixed, each arm error has the exact
pairwise decomposition
\begin{align*}
\delta_{y,n}
&=
\sum_c\frac{w_c}{Q_{yc}}\{M_{yc}-m_{yc}\}\\
&=
\sum_d\sum_c\frac{w_c}{Q_{yc}}
\{M_{dyc}-m_{dyc}\}.
\end{align*}
Substituting these decompositions into the linear term of the Taylor
expansion gives
\begin{align*}
\frac{\delta_{1,n}}{R_1}
-\frac{\delta_{0,n}}{R_0}
&=
\sum_d
\left(
\sum_c\frac{w_c\{M_{d1c}-m_{d1c}\}}{Q_{1c}R_1}
-
\sum_c\frac{w_c\{M_{d0c}-m_{d0c}\}}{Q_{0c}R_0}
\right).
\end{align*}

This expression motivates defining the centered scaled arm and
log-risk-ratio contributions
\[
u_{y,d}
=m_n\sum_c\frac{w_c}{Q_{yc}}\{M_{dyc}-m_{dyc}\},
\qquad
S_{n,d}=\frac{u_{1,d}}{R_1}-\frac{u_{0,d}}{R_0}.
\]
With these definitions, the Taylor linear term is
\[
\frac{\delta_{1,n}}{R_1}
-\frac{\delta_{0,n}}{R_0}
=\frac{1}{m_n}\sum_dS_{n,d}.
\]
Let
\[
\sigma_n^2
=\operatorname{Var}\!\left(\sum_dS_{n,d}\mid\mathcal F_n\right).
\]

To state the conditions controlling the linear term and remainder, write
$d\sim e$ when pairs share an actor, including $d=e$, and define
\[
\Delta_n
=1+\max_d|\{e\ne d:e\sim d\}|.
\]
Uniform positivity means that, for some $\epsilon>0$ and all sufficiently
large $n$,
\[
\min_{y,c}\frac{Q_{yc}}{m_n}\geq\epsilon,
\qquad
\min_yR_y\geq\epsilon.
\]
Together with bounded pair size and finite $\mathcal C$, these inequalities
make the scaled pair scores $S_{n,d}$ uniformly bounded. We require $\sigma_n^2>0$
eventually. The pair scores form an actor-sharing dependency graph on
$\mathcal D_n$: nonadjacent pairs share no actor, so
Assumption~\ref{assump:actor-sharing-sampling} makes them independent given
$\mathcal F_n$. Theorem~\ref{thm:dependency-graph-clt} below applies
\cite[Theorem~2]{janson_1988} to $\sum_d S_{n,d}$ on that graph. The
remaining rate conditions are, for some integer $k\ge3$,
\begin{equation}\label{eq:dyadic-rate-conditions-new}
\frac{\Delta_n}{m_n}\to0,\qquad
\frac{\sigma_n}{m_n}\to0,\qquad
\left(\frac{m_n}{\Delta_n}\right)^{1/k}
\frac{\Delta_n}{\sigma_n}\to0,\qquad
\frac{m_n\Delta_n^3}{\sigma_n^4}\to0.
\end{equation}
The first rate and the third rate, together with bounded summands, are the
Janson hypotheses in the form of
\cite[Condition 2.1]{tabord_meehan_2019}. The second rate is used for the
logarithm remainder below. The fourth rate is used only later, for
variance estimation.
These rate conditions hold in a bounded-degree regime $\Delta_n=O(1)$,
$m_n\to\infty$, $\sigma_n^2\asymp m_n$, and in a dense regime
$m_n\asymp n^2$, $\Delta_n\asymp n$, $\sigma_n^2\asymp n^3$, where
$a_n\asymp b_n$ means that $a_n/b_n$ is bounded away from $0$ and
$\infty$.

Controlling the arm errors requires one additional covariance condition on
the fixed-design cell totals.

\begin{assumption}[Cell-total covariance bound]
\label{assump:cell-total-covariance}
Let
\[
\boldsymbol\xi_n
=
(M_{yc}-m_{yc})_{y\in\{0,1\},c\in\mathcal C}
\]
denote the vector of centered cell totals in the realized design. There
exists a constant $C_\xi<\infty$, independent of $n$, such that eventually
\[
\lambda_{\max}\!\Big(
\operatorname{Cov}(\boldsymbol\xi_n\mid\mathcal F_n)
\Big)
\le C_\xi\sigma_n^2.
\]
\end{assumption}

\begin{lemma}[Arm and cell-rate errors]
\label{lem:arm-error-rate}
Under uniform positivity and
Assumption~\ref{assump:cell-total-covariance},
\[
\delta_{y,n}=O_p\!\left(\frac{\sigma_n}{m_n}\right),
\qquad
\widehat r_{yc}-r_{yc}=O_p\!\left(\frac{\sigma_n}{m_n}\right),
\]
jointly over $y\in\{0,1\}$ and $c\in\mathcal C$.
\end{lemma}

\begin{proof}
Fix $y\in\{0,1\}$. The decomposition above gives
\[
\delta_{y,n}
=
\sum_{c\in\mathcal C}\frac{w_c}{Q_{yc}}\{M_{yc}-m_{yc}\}
=
\mathbf a_{y,n}^\top\boldsymbol\xi_{y,n},
\]
where
\[
\mathbf a_{y,n}
=
\Big(\frac{w_c}{Q_{yc}}\Big)_{c\in\mathcal C},
\qquad
\boldsymbol\xi_{y,n}
=
(M_{yc}-m_{yc})_{c\in\mathcal C}.
\]
Uniform positivity and bounded weights imply
$\|\mathbf a_{y,n}\|^2\le C_a/m_n^2$ for a constant $C_a$ that does not
depend on $n$. Hence
\begin{align*}
\operatorname{Var}(\delta_{y,n}\mid\mathcal F_n)
&=
\mathbf a_{y,n}^\top
\operatorname{Cov}(\boldsymbol\xi_{y,n}\mid\mathcal F_n)
\mathbf a_{y,n}\\
&\le
\|\mathbf a_{y,n}\|^2
\lambda_{\max}\!\Big(
\operatorname{Cov}(\boldsymbol\xi_n\mid\mathcal F_n)
\Big)\\
&\le
\frac{C_aC_\xi\sigma_n^2}{m_n^2}.
\end{align*}
Chebyshev's inequality therefore yields
$\delta_{y,n}=O_p(\sigma_n/m_n)$.

For each $c\in\mathcal C$,
\[
\widehat r_{yc}-r_{yc}
=
\frac{M_{yc}-r_{yc}Q_{yc}}{Q_{yc}}
=
\frac{M_{yc}-m_{yc}}{Q_{yc}},
\]
because $r_{yc}=m_{yc}/Q_{yc}$ under the fixed design. The entry
$M_{yc}-m_{yc}$ is one coordinate of $\boldsymbol\xi_n$, so its conditional
variance is at most
$\lambda_{\max}(\operatorname{Cov}(\boldsymbol\xi_n\mid\mathcal F_n))
\le C_\xi\sigma_n^2$. Dividing by $Q_{yc}\ge\epsilon m_n$ gives
\[
\operatorname{Var}(\widehat r_{yc}-r_{yc}\mid\mathcal F_n)
\le
\frac{C_\xi\sigma_n^2}{\epsilon^2m_n^2},
\]
and another application of Chebyshev yields
$\widehat r_{yc}-r_{yc}=O_p(\sigma_n/m_n)$.
\end{proof}

\begin{corollary}[Fixed-design log-risk-ratio linearization]
\label{cor:log-rr-linearization}
Under Lemma~\ref{lem:arm-error-rate} and the second rate in
\eqref{eq:dyadic-rate-conditions-new},
\begin{equation}\label{eq:fixed-design-linearization}
\widehat\ell_n-\ell_n
=\frac{1}{m_n}\sum_dS_{n,d}
+o_p\!\left(\frac{\sigma_n}{m_n}\right).
\end{equation}
\end{corollary}

\begin{proof}
Lemma~\ref{lem:arm-error-rate} gives
$\delta_{y,n}=O_p(\sigma_n/m_n)$ for each arm. Because
$\sigma_n/m_n\to0$ and $R_y\geq\epsilon$, with probability tending to one
both intermediate values $\widetilde R_y$ in the Taylor expansion are
bounded below by $\epsilon/2$. Consequently,
\[
|\mathcal R_n|
=
O_p\!\left(\frac{\sigma_n^2}{m_n^2}\right)
=
o_p\!\left(\frac{\sigma_n}{m_n}\right).
\]
The result follows on combining this remainder bound with the exact identity
\[
\frac{\delta_{1,n}}{R_1}
-\frac{\delta_{0,n}}{R_0}
=\frac{1}{m_n}\sum_dS_{n,d}.
\]
\end{proof}

\begin{theorem}[Conditional dependency-graph central limit theorem]
\label{thm:dependency-graph-clt}
Under Assumption~\ref{assump:actor-sharing-sampling}, uniform positivity,
bounded pair size, finite $\mathcal C$, and $\sigma_n^2>0$ eventually,
suppose that for some integer $k\ge3$ the first and third rates in
\eqref{eq:dyadic-rate-conditions-new} hold. Then
\[
\frac{1}{\sigma_n}\sum_{d\in\mathcal D_n}S_{n,d}
\ \xrightarrow{d}\ N(0,1)
\qquad\text{conditionally on }\mathcal F_n.
\]
\end{theorem}

\begin{proof}
The centered scores $\{S_{n,d}:d\in\mathcal D_n\}$ are uniformly bounded
under uniform positivity and bounded pair size. Pairs that share no actor
are conditionally independent given $\mathcal F_n$ by
Assumption~\ref{assump:actor-sharing-sampling}, so they form a dependency
graph with maximum degree $\Delta_n$. Because
$E(S_{n,d}\mid\mathcal F_n)=0$ and
$\operatorname{Var}(\sum_dS_{n,d}\mid\mathcal F_n)=\sigma_n^2$,
\cite[Theorem~2]{janson_1988} applies to
$\sum_dS_{n,d}/\sigma_n$ once the first and third rates in
\eqref{eq:dyadic-rate-conditions-new} are imposed, as in
\cite[Propositions 3.1 and 3.2]{tabord_meehan_2019}.
\end{proof}

\begin{corollary}[Fixed-design log-risk-ratio central limit theorem]
\label{cor:fixed-design-clt}
Under Corollary~\ref{cor:log-rr-linearization} and
Theorem~\ref{thm:dependency-graph-clt},
\begin{equation}\label{eq:fixed-design-clt}
\frac{\widehat\ell_n-\ell_n}{\sigma_n/m_n}
\ \xrightarrow{d}\ N(0,1)
\qquad\text{conditionally on }\mathcal F_n.
\end{equation}
\end{corollary}

\begin{proof}
Corollary~\ref{cor:log-rr-linearization} gives
\[
\widehat\ell_n-\ell_n
=\frac{1}{m_n}\sum_dS_{n,d}
+o_p\!\left(\frac{\sigma_n}{m_n}\right).
\]
Theorem~\ref{thm:dependency-graph-clt} gives
\[
\frac{1}{\sigma_n/m_n}\cdot\frac{1}{m_n}\sum_dS_{n,d}
=
\frac{1}{\sigma_n}\sum_dS_{n,d}
\ \xrightarrow{d}\ N(0,1).
\]
Slutsky's theorem completes the proof.
\end{proof}

Next, we consider studentization. The empirical fixed-design arm and pair
scores are
\[
\widehat u_{y,d}
=m_n\sum_c\frac{w_c}{Q_{yc}}
\{M_{dyc}-\widehat r_{yc}Q_{dyc}\},
\qquad
\widehat S_{n,d}
=\frac{\widehat u_{1,d}}{\widehat R_1}
-\frac{\widehat u_{0,d}}{\widehat R_0},
\]
and $\widehat\phi_d=\widehat S_{n,d}/m_n$. These scores sum to zero.
To expose the centering problem, define
\[
b_{y,d}
=m_n\sum_c\frac{w_c}{Q_{yc}}\{m_{dyc}-r_{yc}Q_{dyc}\},
\qquad
a_{n,d}=\frac{b_{1,d}}{R_1}-\frac{b_{0,d}}{R_0}.
\]
Although $\sum_da_{n,d}=0$, $a_{n,d}$ need not vanish pair by pair.

Variance estimation likewise aggregates pair scores to actors because the
dependency graph is indexed by actor overlap: pairs are adjacent when they
share an actor, and conditionally independent otherwise.
Let $B_n$ be the actor--pair incidence matrix,
$(B_n)_{g,d}=I(g\in d)$, and put
\[
K_n=B_n^\top B_n-I_{m_n}.
\]
Thus $(K_n)_{de}=I(d\sim e)$. The inclusion--exclusion form is the
dyadic-robust quadratic form of \cite{tabord_meehan_2019}. In scaled
notation, the inclusion--exclusion and actor-sum quadratic forms are
\[
\widehat V_n^{\mathrm{IE}}
=\widehat{\mathbf S}_n^\top K_n\widehat{\mathbf S}_n,
\qquad
\widehat V_n^{\mathrm{AS}}
=\widehat{\mathbf S}_n^\top B_n^\top B_n\widehat{\mathbf S}_n
=\sum_g\left(\sum_{d:g\in d}\widehat S_{n,d}\right)^2.
\]
Their corresponding log-risk-ratio variances divide by $m_n^2$.

\begin{lemma}[Empirical-score variance expansions]
\label{lem:empirical-score-variance-expansions}
Under Assumptions~\ref{assump:actor-sharing-sampling} and
\ref{assump:cell-total-covariance}, bounded pair size, finite $\mathcal C$,
and the rates in \eqref{eq:dyadic-rate-conditions-new}, uniformly in $d$,
\begin{equation}\label{eq:empirical-score-approx}
\widehat S_{n,d}=S_{n,d}+a_{n,d}
+O_p\!\left(\frac{\sigma_n}{m_n}\right).
\end{equation}
The variance estimators satisfy
\begin{align}
\widehat V_n^{\mathrm{IE}}
&=\sigma_n^2+\mathbf a_n^\top K_n\mathbf a_n+o_p(\sigma_n^2),
\label{eq:ie-expansion}\\
\widehat V_n^{\mathrm{AS}}
&=\sigma_n^2+\sum_dE(S_{n,d}^2\mid\mathcal F_n)
+\mathbf a_n^\top B_n^\top B_n\mathbf a_n
+o_p(\sigma_n^2).
\label{eq:as-expansion}
\end{align}
\end{lemma}

\begin{proof}
The difference \(\widehat S_{n,d}-S_{n,d}-a_{n,d}\) is the change from
replacing each \(r_{yc}\) by \(\widehat r_{yc}\) and each \(R_y\) by
\(\widehat R_y\). Lemma~\ref{lem:arm-error-rate} gives
\(\widehat r_{yc}-r_{yc}=O_p(\sigma_n/m_n)\) and
\(\delta_{y,n}=O_p(\sigma_n/m_n)\). Pair-level cell counts are bounded and
\(m_n/Q_{yc}=O(1)\), so this change is \(O_p(\sigma_n/m_n)\) uniformly in
\(d\). This is \eqref{eq:empirical-score-approx}.

Write \(\widehat{\mathbf S}_n=\mathbf S_n+\mathbf a_n+\mathbf e_n\) with
\(\|\mathbf e_n\|_\infty=O_p(\sigma_n/m_n)\). Each row of \(K_n\) has at most
\(\Delta_n\) nonzero entries, so
\[
\widehat{\mathbf S}_n^\top K_n\widehat{\mathbf S}_n
-(\mathbf S_n+\mathbf a_n)^\top K_n(\mathbf S_n+\mathbf a_n)
=o_p(\sigma_n^2),
\]
because the leading error ratios are \(O_p(\Delta_n/\sigma_n)\) and
\(O_p(\Delta_n/m_n)\). The same bound holds with \(B_n^\top B_n=K_n+I_{m_n}\)
in place of \(K_n\).

Nonadjacent pairs are conditionally independent given \(\mathcal F_n\), so
\[
E(\mathbf S_n^\top K_n\mathbf S_n\mid\mathcal F_n)=\sigma_n^2.
\]
The form \(\mathbf S_n^\top K_n\mathbf S_n\) is a sum of \(O(m_n\Delta_n)\)
bounded products, and each product is dependent on at most \(O(\Delta_n^2)\)
others. The last rate in \eqref{eq:dyadic-rate-conditions-new} therefore
gives
\[
\operatorname{Var}(\mathbf S_n^\top K_n\mathbf S_n\mid\mathcal F_n)
=O(m_n\Delta_n^3)=o(\sigma_n^4).
\]
Chebyshev's inequality yields
\(\mathbf S_n^\top K_n\mathbf S_n=\sigma_n^2+o_p(\sigma_n^2)\). The
coefficients of \(K_n\mathbf a_n\) are \(O(\Delta_n)\), and the same
counting gives \(\mathbf a_n^\top K_n\mathbf S_n=o_p(\sigma_n^2)\). Expanding
\[
(\mathbf S_n+\mathbf a_n)^\top K_n(\mathbf S_n+\mathbf a_n)
=\mathbf S_n^\top K_n\mathbf S_n
+2\mathbf a_n^\top K_n\mathbf S_n
+\mathbf a_n^\top K_n\mathbf a_n
\]
proves \eqref{eq:ie-expansion}.

Because \(B_n^\top B_n=K_n+I_{m_n}\),
\[
E(\mathbf S_n^\top B_n^\top B_n\mathbf S_n\mid\mathcal F_n)
=\sigma_n^2+\sum_dE(S_{n,d}^2\mid\mathcal F_n),
\]
and the analogous remainder and cross term are again \(o_p(\sigma_n^2)\).
This is \eqref{eq:as-expansion}. The mean-square counting is that of
\cite[Propositions 3.3--3.5]{tabord_meehan_2019}.
\end{proof}

The term $\mathbf a_n^\top K_n\mathbf a_n$ has no fixed sign. For example,
three pair bundles forming a path have
\[
K=
\begin{pmatrix}
1&1&0\\
1&1&1\\
0&1&1
\end{pmatrix},
\qquad
\mathbf a=(1,-2,1)^\top,
\qquad
\mathbf a^\top K\mathbf a=-2.
\]
Repeating suitably scaled copies gives bounded triangular arrays for which the
inclusion--exclusion estimator is asymptotically anti-conservative. Exact
studentization therefore requires the additional condition
\begin{equation}\tag{MH}\label{eq:mean-heterogeneity-condition}
\frac{|\mathbf a_n^\top K_n\mathbf a_n|}{\sigma_n^2}\longrightarrow0.
\end{equation}
By contrast, $B_n^\top B_n$ is positive semidefinite, so both additional
terms in \eqref{eq:as-expansion} are nonnegative.

\begin{lemma}[Fixed-design studentization]
\label{lem:fixed-design-studentization}
Suppose the conditions of Corollary~\ref{cor:fixed-design-clt} and
Lemma~\ref{lem:empirical-score-variance-expansions} hold. Then, conditionally
on $\mathcal F_n$, the following statements hold.
\begin{enumerate}
\item If \eqref{eq:mean-heterogeneity-condition} holds, then
\[
\frac{\widehat V_n^{\mathrm{IE}}}{\sigma_n^2}\xrightarrow{p}1,
\qquad
\frac{\widehat\ell_n-\ell_n}
{\sqrt{\widehat V_n^{\mathrm{IE}}}/m_n}
\xrightarrow{d}N(0,1).
\]
\item Without \eqref{eq:mean-heterogeneity-condition},
\[
\frac{\widehat V_n^{\mathrm{AS}}}{\sigma_n^2}\geq1+o_p(1),
\]
and, for every $0<\alpha<1/2$,
\[
\liminf_{n\to\infty}
P\!\left\{
\ell_n\geq
\widehat\ell_n-z_{1-\alpha}
\frac{\sqrt{\widehat V_n^{\mathrm{AS}}}}{m_n}
\;\middle|\;\mathcal F_n
\right\}
\geq1-\alpha.
\]
\end{enumerate}
\end{lemma}

\begin{proof}
Under \eqref{eq:mean-heterogeneity-condition},
\eqref{eq:ie-expansion} gives
$\widehat V_n^{\mathrm{IE}}/\sigma_n^2\to_p1$. The first studentized limit
then follows from Corollary~\ref{cor:fixed-design-clt} and Slutsky's theorem.
This is the dyadic-robust studentization of
\cite[Theorem 3.1]{tabord_meehan_2019}, applied to the linearized
log-risk-ratio scores once
\eqref{eq:mean-heterogeneity-condition} removes the extra term.

For the actor-sum estimator, which is a conservative modification of the
dyadic-robust form and is not claimed by \cite{tabord_meehan_2019},
\eqref{eq:as-expansion} and positive
semidefiniteness of $B_n^\top B_n$ give
$\widehat V_n^{\mathrm{AS}}/\sigma_n^2\geq1+o_p(1)$. Let
\[
Z_n=\frac{\widehat\ell_n-\ell_n}{\sigma_n/m_n},
\qquad
c_n=\frac{\sqrt{\widehat V_n^{\mathrm{AS}}}}{\sigma_n}.
\]
Corollary~\ref{cor:fixed-design-clt} gives
$Z_n\xrightarrow{d}N(0,1)$, while the variance inequality gives
$c_n\geq1-o_p(1)$. For every fixed $\eta\in(0,1)$, the event
\[
\{Z_n\leq z_{1-\alpha}(1-\eta),\ c_n\geq1-\eta\}
\]
is contained in $\{Z_n\leq z_{1-\alpha}c_n\}$ with probability tending to
one. Taking the limit inferior and then letting $\eta\downarrow0$ proves the
one-sided coverage claim.
\end{proof}

The implementation uses
\[
\widehat\sigma_{\log\mathrm{RR},\mathrm{AS}}^2
=\widehat V_n^{\mathrm{AS}}/m_n^2
=\sum_g\widehat\Phi_g^2.
\]
Equivalently, the actor-sum variance is the inclusion--exclusion variance
plus a sum of squares:
\[
\widehat\sigma_{\log\mathrm{RR},\mathrm{AS}}^2
-\widehat\sigma_{\log\mathrm{RR},\mathrm{IE}}^2
=\sum_d\widehat\phi_d^2\ge0.
\]
\begin{proof}[Proof of Theorem~\ref{cor:dyadic-lower-limit}]
Lemma~\ref{lem:fixed-design-studentization} supplies a conservative
conditional lower limit for $\ell_n$ and hence, after exponentiation, for
$\text{RR}^{\mathrm{obs}}_n(\mathcal F_n)$. Under
Assumption~\ref{assump:finite-frame-bridge}, the common-weight argument in
Section~\ref{sec:cde-estimation} applies to the latent-type distributions and
conditional potential risks in the fixed frame and gives
\[
\text{CRR}^{\mathrm{conn}}_n(\mathcal F_n)
\ge
\frac{\text{RR}^{\mathrm{obs}}_n(\mathcal F_n)}
{BF_{\text{CDE,conn},n}}.
\]
Divide the observed-risk-ratio lower limit by the fixed factor $BF$ and apply
this deterministic inequality.
\end{proof}

\begin{proof}[Proof of Corollary~\ref{cor:dyadic-fc-lower-limit}]
The stochastic argument in Theorem~\ref{cor:dyadic-lower-limit} is unchanged:
it gives a conservative conditional lower limit for
$\text{RR}^{\mathrm{obs}}_n(\mathcal F_n)$. Assumptions
\ref{assump:finite-frame-bridge} and
\ref{assump:finite-frame-fc-transport} give
\[
\text{CRR}^{\mathrm{fc}}_n(\mathcal F_n)
\geq
\frac{\text{RR}^{\mathrm{obs}}_n(\mathcal F_n)}
{BF_{\text{CDE,fc},n}}.
\]
Because $BF\geq BF_{\text{CDE,fc},n}$, the observed-risk coverage event is
contained in
$\{\text{CRR}^{\mathrm{fc}}_n(\mathcal F_n)
\geq\underline L_{1-\alpha}^{\mathrm{AS}}(BF)\}$, which proves the result.
\end{proof}

\subsection{Ego-cluster bootstrap}\label{app:ego-bootstrap-proof}

\begin{proof}[Proof of Theorem~\ref{cor:ego-bootstrap-lower-limit}]
All probabilities and expectations in this proof condition on the fixed
eligible ego frame. Set
\[
X_{n,i}=a_n^{-1}T_{n,i},\qquad
\bar X_n=n_E^{-1}\sum_{i\in\mathcal E_n}X_{n,i},\qquad
\mu_n=E(X_{n,i}),\qquad
\Sigma_n=\operatorname{Var}(X_{n,i}).
\]
The common-weight log risk ratio is a smooth, scale-invariant function $h$
of the cluster means, so $\widehat\ell=h(\bar X_n)$ and the population
centering is $\ell_n=h(\mu_n)$. On samples outside the positive domain of
$h$, assign the log estimator any fixed finite value, and use the same
convention for bootstrap samples, with the risk ratio defined as its
exponential. The probabilities of these events vanish,
as shown below, so this convention does not affect the limit.

Assumption~\ref{assump:ego-cluster-regularity} gives $\mu_n\to\mu$,
$\Sigma_n\to\Sigma$, and uniformly bounded centered $(2+\delta)$ moments.
The Lindeberg condition follows from
\[
E\!\left[
\|X_{n,i}-\mu_n\|^2
I\{\|X_{n,i}-\mu_n\|>\epsilon\sqrt{n_E}\}
\right]
\leq
\frac{E\|X_{n,i}-\mu_n\|^{2+\delta}}
{(\epsilon\sqrt{n_E})^\delta}\longrightarrow0
\]
for every $\epsilon>0$. The Lindeberg--Feller central limit theorem
\cite[Theorem~2.27]{van_der_vaart_1998}, applied coordinatewise via the
Cram\'er--Wold device, and the delta method
\cite[Theorem~3.1]{van_der_vaart_1998} therefore give
\[
\sqrt{n_E}(\bar X_n-\mu_n)\xrightarrow{d}N(0,\Sigma),
\qquad
\sqrt{n_E}(\widehat\ell-\ell_n)
\xrightarrow{d}N(0,v),
\]
where $v=\nabla h(\mu)^\top\Sigma\nabla h(\mu)>0$.
Here the Taylor expansion is around $\mu_n$, with
$\nabla h(\mu_n)\to\nabla h(\mu)$; it does not replace $h(\mu_n)$ by
$h(\mu)$ in the centering.

The empirical bootstrap draws $n_E$ eligible ego blocks with replacement,
including any zero-total blocks. Let $\bar X_n^\star$ denote their mean.
The empirical covariance converges in probability to $\Sigma$ by the
triangular-array weak law, using the uniform $(2+\delta)$ moment bound.
That bound also gives
$n_E^{-1}\sum_i\|X_{n,i}-\bar X_n\|^{2+\delta}=O_p(1)$, so the
conditional Lindeberg term is at most
\[
\frac{1}{(\epsilon\sqrt{n_E})^\delta}
\frac{1}{n_E}\sum_{i\in\mathcal E_n}
\|X_{n,i}-\bar X_n\|^{2+\delta}=o_p(1).
\]
Consequently, the conditional bootstrap central limit theorem
\citep{field_welsh_2007} and the delta method
\cite[Theorem~3.1]{van_der_vaart_1998} yield
\[
\sup_x\left|
P^\star\!\left\{\sqrt{n_E}
(\widehat\ell^\star-\widehat\ell)\leq x\right\}
-\Phi(x/\sqrt v)
\right|\xrightarrow{p}0.
\]
Continuity of the nondegenerate Gaussian distribution makes this convergence
uniform in $x$. Also, $\bar X_n\to_p\mu$ and
$\bar X_n^\star-\bar X_n=o_{P^\star}(1)$ in probability. Since $\mu$ is
in the positive domain of $h$, the original sample leaves that domain with
probability $o(1)$ and the conditional bootstrap probability of doing so is
$o_p(1)$. This justifies the convention above. Discarding unsupported
bootstrap draws when defining the percentile gives the same first-order limit.

Write $z_\alpha=\Phi^{-1}(\alpha)$. The exponential map is strictly
increasing, so quantile convergence gives
\[
c_{\alpha,n}^\star
:=\sqrt{n_E}\{\log q_\alpha^\star-\widehat\ell\}
=\sqrt v\,z_\alpha+o_p(1).
\]
It follows that
\[
\begin{split}
P\!\left\{\text{RR}^{\mathrm{obs}}_{\mathrm{conn},n}
\geq q_\alpha^\star\right\}
&=P\!\left\{\sqrt{n_E}(\widehat\ell-\ell_n)
\leq-c_{\alpha,n}^\star\right\}\\
&\longrightarrow\Phi(-z_\alpha)=1-\alpha.
\end{split}
\]

Under the population law at each sufficiently large $n$,
Assumptions~\ref{assump:causal-bridge},
\ref{assump:background-mean-bridge} and the common-weight inequality in
Section~\ref{sec:cde-estimation} give
\[
\text{CRR}^{\text{conn}}_{\text{CDE},n}
\ge
\frac{\text{RR}^{\text{obs}}_{\text{conn},n}}{BF_{\text{CDE,conn},n}}.
\]
For fixed $BF\geq BF_{\text{CDE,conn},n}$, the observed-risk coverage
event is contained in
$\{\text{CRR}^{\text{conn}}_{\text{CDE},n}\geq q_\alpha^\star/BF\}$.
Taking the limit inferior proves the claim. If a fixed observed target
$\ell$ instead satisfies $\sqrt{n_E}(\ell_n-\ell)\to0$, Slutsky's theorem
gives the same centered limit at $\ell$, followed by the corresponding
fixed-population sensitivity inequality.
\end{proof}

\begin{proof}[Proof of Corollary~\ref{cor:ego-fc-bootstrap-lower-limit}]
The bootstrap argument in Theorem~\ref{cor:ego-bootstrap-lower-limit} gives
first-order lower coverage of $\text{RR}^{\mathrm{obs}}_{\text{conn},n}$ by
$q_\alpha^\star$. Under Assumptions~\ref{assump:causal-bridge},
\ref{assump:background-mean-bridge}, and
\ref{assump:forced-contact-transport}, together with the forced-contact
sensitivity restrictions under the population law at each sufficiently large
$n$,
\[
\text{CRR}^{\mathrm{fc}}_{\text{CDE},n}
\geq
\frac{\text{RR}^{\mathrm{obs}}_{\text{conn},n}}
{BF_{\text{CDE,fc},n}}.
\]
Division by fixed $BF\geq BF_{\text{CDE,fc},n}$ gives the claimed event
inclusion and hence the corollary.
\end{proof}

\section{Oracle CDE targets}\label{app:simulation-oracle-cde}

The simulation targets both controlled direct effect risk ratios in Section
\ref{sec:rr-target-overview}. For each ego-alter pair $(i,j)$, we define the
oracle intervention by setting $A_{ij}=1$ and setting $Y_j^{t-1}$ to either 1 or
0, then recomputing the ego's neighbor mean under the node-level outcome model.
Thus, for each focal dyad, the intervention changes the contribution of alter
$j$ to ego $i$'s neighborhood exposure while leaving the rest of the network and
the ego's baseline state fixed.

This intervention is intentionally dyad-by-dyad. It does not require all
possible ties to be simultaneously present in the realized network. Rather, it
asks for the risk that would be induced for ego $i$ if this particular alter
were placed in contact and assigned a prior outcome value. The effect can vary
across egos because the same intervention changes the neighbor mean by different
amounts depending on the ego's current degree. This is a feature of the
node-level contagion model rather than a nuisance: highly connected egos are
less affected by changing a single alter, while low-degree egos can be more
strongly affected.

The implemented oracle CDE is computed on the probability scale from the known
DGP parameters. For a dyad-transition row $(i,j,t)$, let
\[
S_i^{t-1}=\sum_{\ell\neq i} A_{i\ell}Y_\ell^{t-1},
\qquad
d_i=\sum_{\ell\neq i} A_{i\ell}.
\]
Under the intervention setting $A_{ij}=1$ and $Y_j^{t-1}=y$, the intervened
degree and neighbor-outcome sum are
\[
d_{i,ij}^{\star}=d_i+(1-A_{ij}),
\qquad
S_{i,ij}^{\star}(y)=S_i^{t-1}-A_{ij}Y_j^{t-1}+y.
\]
The corresponding intervened neighbor mean is
\[
\bar{Y}_{N_i,ij}^{\star}(y)
=
S_{i,ij}^{\star}(y)/d_{i,ij}^{\star}.
\]
Because the focal edge is set to one, $d_{i,ij}^{\star}\geq 1$ for all ordered
dyads. The DGP-implied interventional risk for that dyad-transition row is
\[
p_{ij}^{t,\star}(y)
=
\operatorname{expit}
\left\{
\alpha_Y+\beta_U U_i+\beta_PY_i^{t-1}
+\beta_N\bar{Y}_{N_i,ij}^{\star}(y)
\right\}.
\]

The oracle targets directly average these DGP-implied focal-component risks
over their respective background frames. Let
$\mathcal D_{A,c}=\{(i,j,t):A_{ij}=1,Y_i^{t-1}=c\}$,
$\mathcal D_c=\{(i,j,t):Y_i^{t-1}=c\}$, and let $\widehat w(c)$ be the share
of connected rows with $Y_i^{t-1}=c$. As in the estimator, we retain only
$c$ strata containing both focal-exposure arms among connected rows and
renormalize $\widehat w(c)$ over those strata. Define
\[
\mu_{y,c}^{\mathrm{conn}}
=
\frac{1}{|\mathcal D_{A,c}|}
\sum_{(i,j,t)\in\mathcal D_{A,c}}p_{ij}^{t,\star}(y),
\qquad
\mu_{y,c}^{\mathrm{fc}}
=\frac{1}{|\mathcal D_c|}
\sum_{(i,j,t)\in\mathcal D_c}p_{ij}^{t,\star}(y).
\]
The naturally connected oracle is
\begin{align*}
\text{CRR}^{\text{true,conn}}_{\text{CDE}}
&:=
\frac{
\sum_c \widehat w(c)\mu_{1,c}^{\mathrm{conn}}
}{
\sum_c \widehat w(c)\mu_{0,c}^{\mathrm{conn}}
},
\end{align*}
and the forced-contact oracle averages over all ordered-dyad backgrounds,
\begin{align*}
\text{CRR}^{\text{true,fc}}_{\text{CDE}}
&:=
\frac{
\sum_c \widehat w(c)\mu_{1,c}^{\mathrm{fc}}
}{
\sum_c \widehat w(c)\mu_{0,c}^{\mathrm{fc}}
}.
\end{align*}
Latent-type standardization is used to construct and diagnose the sensitivity
factors, not to substitute connected-background risks into the forced-contact
truth.
We compare the observed connected-dyad risk ratio to both oracle targets. For
the naturally connected target, the CDE sensitivity analysis accounts for latent
imbalance among connected dyads. For the forced-contact target, it also accounts
for the gap between naturally connected dyads and the full dyad population that
would be placed in contact under the intervention.

The simulation computes all CDE quantities in Section \ref{sec:rr-cde} as oracle
diagnostics because the latent types and DGP-implied interventional risks are
known. For each simulated panel we report the forced-contact truth
$\text{CRR}^{\text{fc}}_{\text{CDE}}$, the naturally connected truth
$\text{CRR}^{\text{conn}}_{\text{CDE}}$, and the observed connected-dyad risk
ratio $\text{RR}^{\text{obs}}_{\text{conn}}$, the last of which is itself
averaged over the retained connected-dyad distribution of $c=Y_i^{t-1}$ by
combining the within-$c$ exposed and unexposed risks with the weights
$\widehat w(c)$.

The empirical procedure takes one sensitivity factor for each target, valid
uniformly across the retained strata. For simulation diagnostics, the known
DGP supplies a single frame-level oracle factor satisfying that requirement;
its construction from within-$c$ DGP quantities is not a practitioner
specification. The connected-target bound divides
$\text{RR}^{\text{obs}}_{\text{conn}}$ by $BF_{\text{CDE,conn}}$, while the
forced-contact bound uses the corresponding composed factor
$BF_{\text{CDE,fc}}$. We also compute the one-step direct diagnostic
$BF_{\text{CDE,fc}}^{\text{dir}}$, but the plotted main forced-contact bound
uses the composed-ratio factor.

The ego-centric analysis follows Section~\ref{sec:ego-centric-bounds}. Its
oracle calculation analogously supplies one factor valid across eligible ego
types and retained measured strata.

\section{Simulation point-estimate heatmaps}\label{app:simulation-results}

This appendix reports the additional $n=250$ point-estimate heatmaps for the no
and strong latent-susceptibility scenarios. The corresponding moderate scenario
figures are shown in the main text. For each $U\to Y$ scenario, the connected
figure shows the oracle naturally connected CDE truth, observed connected-dyad
risk ratio, pooled connected lower bound, and ego-centric connected lower bound.
The replicate and empirical-support averaging convention is given in
Section~\ref{sec:simulation-study}.

\begin{figure}[H]
\centering
\includegraphics[width=\linewidth]{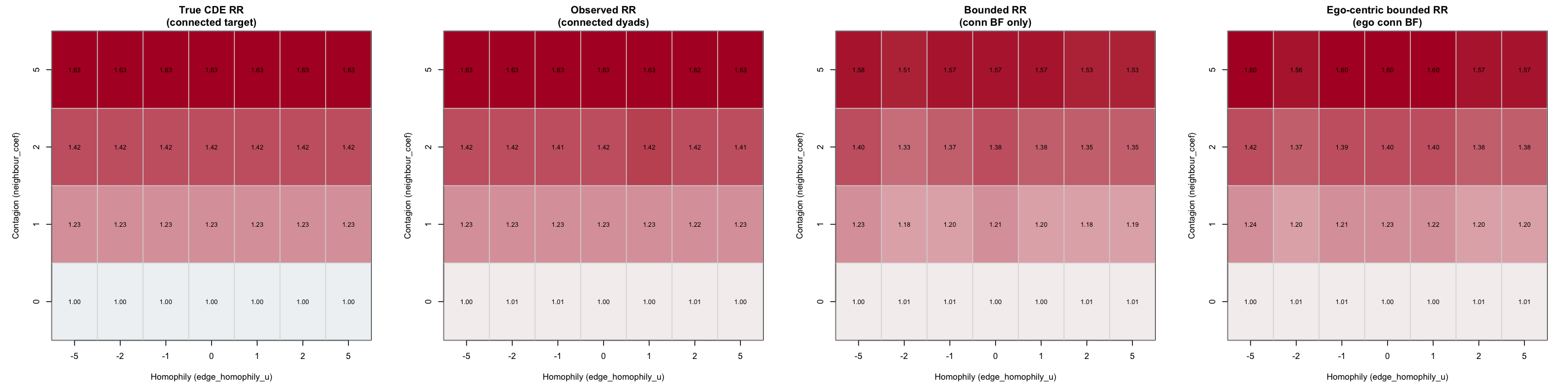}
\caption{Naturally connected CDE simulation heatmaps for no latent effect on
the outcome ($n=250$, $\beta_U=0$), with the edge intercept calibrated to
target mean degree 3. Panels show the oracle connected CDE truth, observed
connected-dyad risk ratio, pooled connected lower bound, and ego-centric
connected lower bound. Within each heatmap, the horizontal axis varies latent
homophily and the vertical axis varies contagion. All panels use a shared
diverging color scale, with white at RR${}=1$.}
\label{fig:sim-connected-four-panel-none}
\end{figure}

\begin{figure}[H]
\centering
\includegraphics[width=\linewidth]{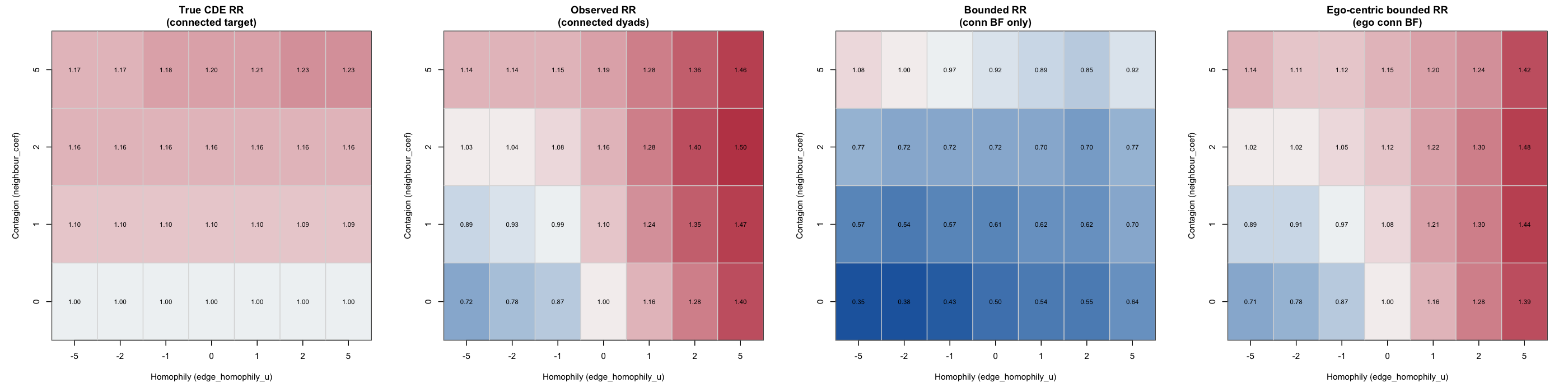}
\caption{Naturally connected CDE simulation heatmaps for a strong latent
effect on the outcome ($n=250$, $\beta_U=2$), with the edge intercept
calibrated to target mean degree 3. Panels show the oracle connected CDE truth,
observed connected-dyad risk ratio, pooled connected lower bound, and
ego-centric connected lower bound. Within each heatmap, the horizontal axis
varies latent homophily and the vertical axis varies contagion. All panels use
a shared diverging color scale, with white at RR${}=1$.}
\label{fig:sim-connected-four-panel-strong}
\end{figure}

\section{One-sided inference evaluation}
\label{app:bootstrap-inference}

The simulation evaluates fixed-design one-sided inference on the naturally
connected target. Each cell uses $n_{\text{truth}}=250$ independent truth
draws and $n_{\text{eval}}=250$ independent evaluation draws, with 250
ego-bootstrap replicates per evaluation draw. We first estimate Monte Carlo
summaries:
\[
CRR_{\text{truth,CDE}}^{\text{conn}}
=
\frac{1}{n_{\text{truth}}}
\sum_{r=1}^{n_{\text{truth}}}
\widehat{CRR}^{\text{true,conn},(r)}_{\text{CDE}},
\qquad
BF_{\text{cell}}
=
\frac{1}{|\mathcal V_{\text{truth}}|}
\sum_{r\in\mathcal V_{\text{truth}}}
\widehat{BF}_{\text{CDE,conn}}^{(r)}.
\]
Here $\mathcal V_{\text{truth}}$ contains truth draws with exact empirical
support for the connected factor. Its size is recorded for every cell.
The cell-level factor $BF_{\text{cell}}$ is then held fixed in every evaluation
replicate. It is a DGP-calibrated sensitivity restriction, not a factor
estimated from the evaluation outcomes.
Because $BF_{\text{cell}}$ averages sample oracle factors, it need not upper
bound every realized-frame factor. We therefore distinguish practical
fixed-factor rejection behavior from formal coverage diagnostics.

For each evaluation design, the known outcome model supplies the exact
conditional observed risk ratio
$\text{RR}^{\text{obs}}_n(\mathcal F_n)$ and the realized-frame connected CDE.
Observed-risk coverage compares the actor-sum lower limit directly with the
former. Causal coverage uses the replicate-specific support-valid oracle
connected factor and compares the resulting limit with the replicate-specific
connected CDE, retaining evaluations whose exact conditional oracle inequality
holds as described in Section~\ref{sec:simulation-study}. Realized outcome
noise can make an estimated point bound cross the CDE even when this oracle
inequality holds; such noisy crossings are not used to exclude evaluations.
This separates sampling calibration from finite Monte Carlo calibration of
$BF_{\text{cell}}$.

In each evaluation replicate we compute the fixed-design score, the primary
positive-semidefinite actor-sum standard error, and the $t_\kappa$
one-sided limit. The raw test uses $BF=1$; the practical sensitivity test
divides by $BF_{\text{cell}}$. We also record the sharper
inclusion--exclusion limit, its loading increment
$\sum_d\widehat\phi_d^2$, and the oracle mean-heterogeneity variance ratio.
The ego-cluster procedure resamples outgoing ego blocks, takes the one-sided
$\alpha$ percentile, and divides by the same fixed cell factor. Nonfinite or
nonpositive risk ratios are discarded; no ego-cluster limit is reported if
fewer than $80\%$ of bootstrap draws remain valid.

A method rejects $\text{CRR}^{\text{conn}}_{\text{CDE}}\le 1$ when its
one-sided lower limit exceeds one. When the cell-level truth is not above one,
that rejection is Type I error. When the truth is above one, failure to reject
is Type II error.

The simulations are intended to be read comparatively. The raw actor-sum test of
the observed risk ratio has the most power, but it is also most vulnerable to
latent homophily. The actor-sum test after dividing by the fixed
connected factor is the primary sensitivity analysis. The ego-cluster percentile is a
stronger-assumption nonparametric check: it is typically close to the actor-sum
limit, but it does not allow residual dependence through a shared alter.

The target-specific truth, error, and power figures show the oracle CDE truth in
the first column, raw actor-sum Type I and Type II error in the second
column, actor-sum-plus-fixed-factor error in the third column, and
ego-cluster-plus-fixed-factor error in the fourth column.

\begin{figure}[H]
\centering
\includegraphics[width=\linewidth]{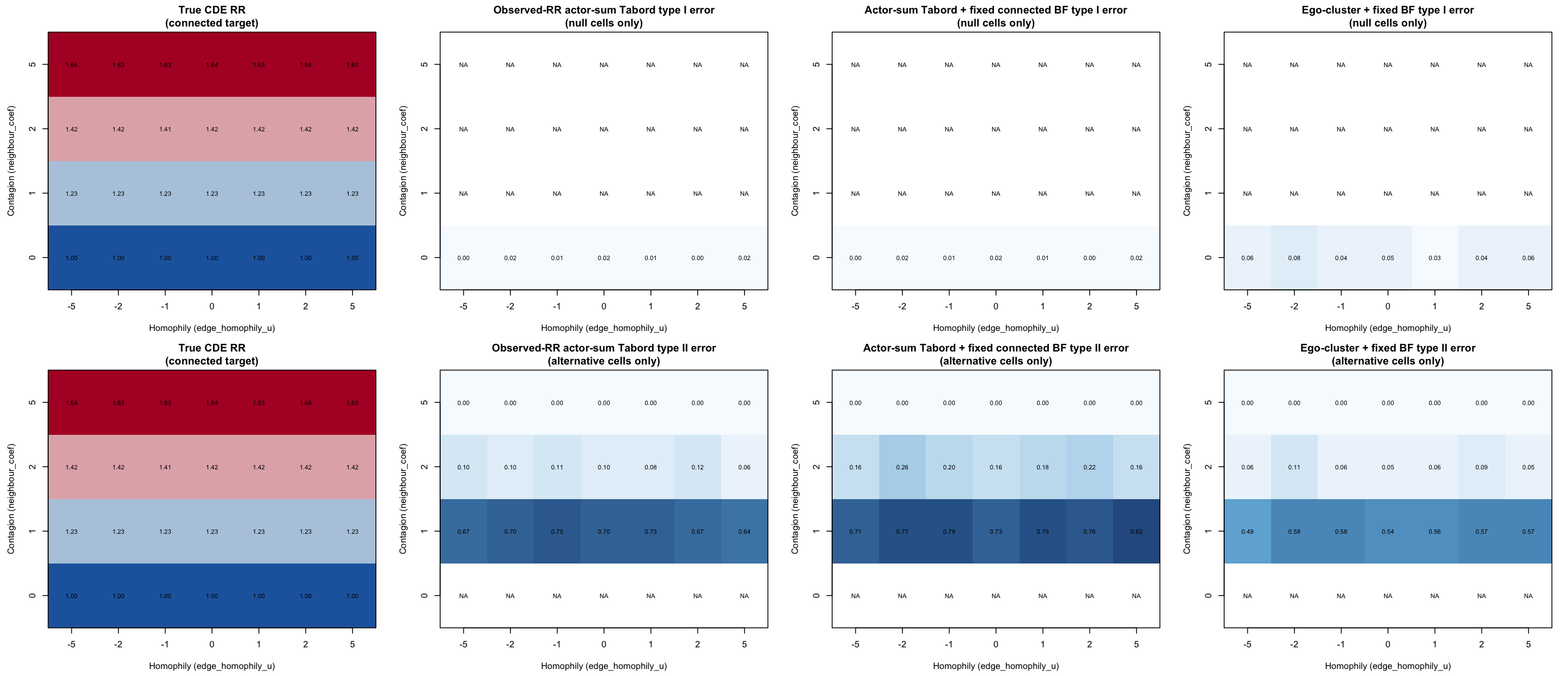}
\caption{Naturally connected CDE truth and one-sided Type I/Type II error for
no latent effect on the outcome ($n=250$, $\beta_U=0$), with the edge
intercept calibrated to target mean degree 3. Panel columns are truth, raw
actor-sum, actor-sum with a fixed connected factor, and the
ego-cluster check. The top row reports Type I error in null cells and the
bottom row reports Type II error in alternative cells. Within each heatmap,
the horizontal axis varies latent homophily and the vertical axis varies
contagion; error panels share a 0--1 color scale.}
\label{fig:sim-error-power-connected-none}
\end{figure}

\begin{figure}[H]
\centering
\includegraphics[width=\linewidth]{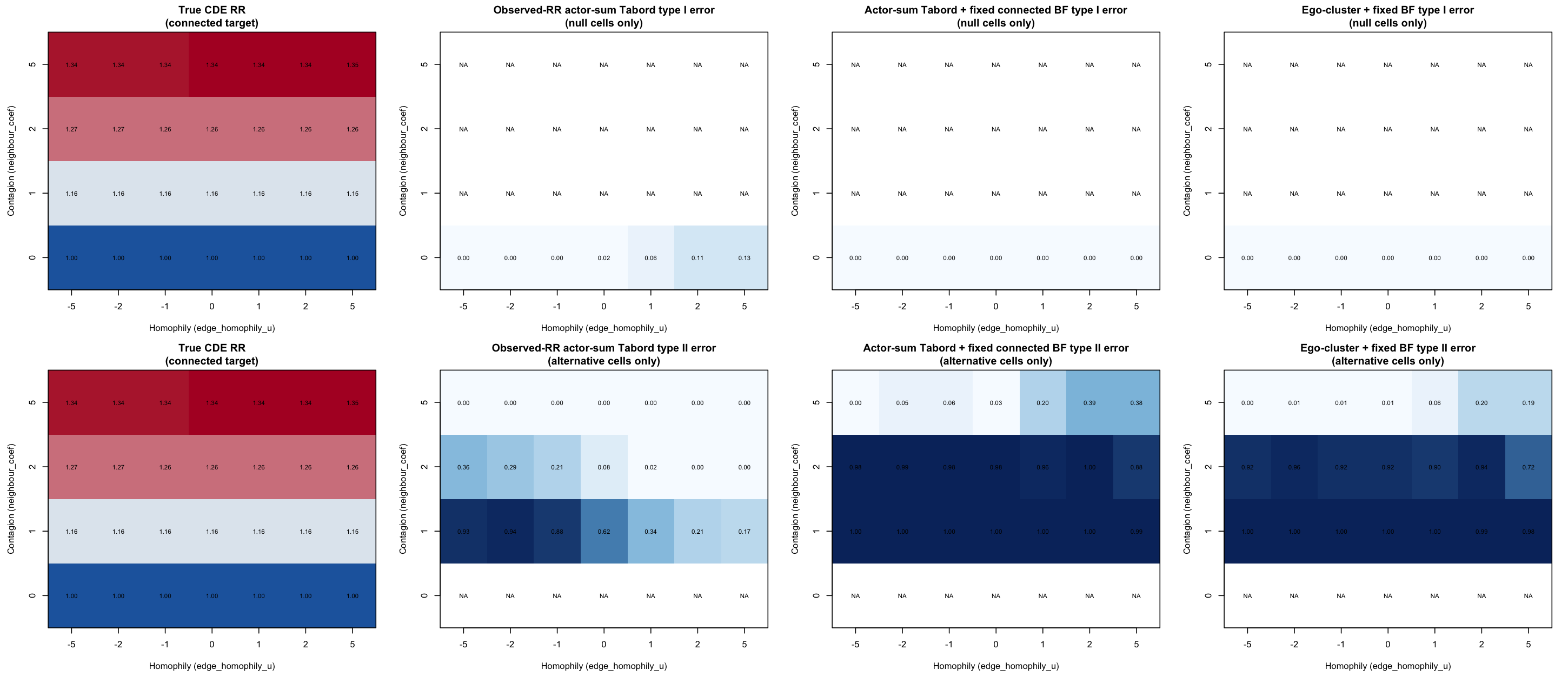}
\caption{Naturally connected CDE truth and one-sided Type I/Type II error for
a moderate latent effect on the outcome ($n=250$, $\beta_U=1$), with the edge
intercept calibrated to target mean degree 3. Panel columns are truth, raw
actor-sum, actor-sum with a fixed connected factor, and the
ego-cluster check. The top row reports Type I error in null cells and the
bottom row reports Type II error in alternative cells. Within each heatmap,
the horizontal axis varies latent homophily and the vertical axis varies
contagion; error panels share a 0--1 color scale.}
\label{fig:sim-error-power-connected-moderate}
\end{figure}

\begin{figure}[H]
\centering
\includegraphics[width=\linewidth]{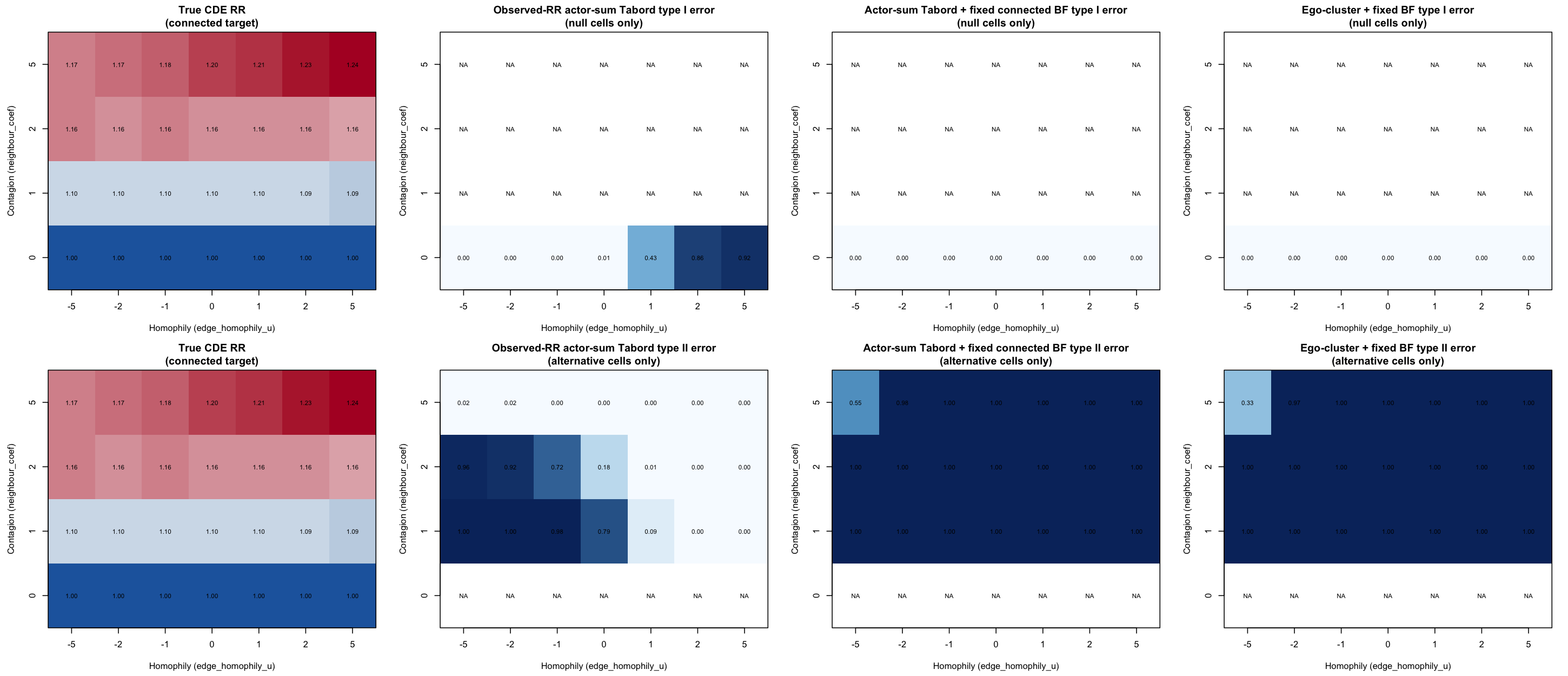}
\caption{Naturally connected CDE truth and one-sided Type I/Type II error for
a strong latent effect on the outcome ($n=250$, $\beta_U=2$), with the edge
intercept calibrated to target mean degree 3. Panel columns are truth, raw
actor-sum, actor-sum with a fixed connected factor, and the
ego-cluster check. The top row reports Type I error in null cells and the
bottom row reports Type II error in alternative cells. Within each heatmap,
the horizontal axis varies latent homophily and the vertical axis varies
contagion; error panels share a 0--1 color scale.}
\label{fig:sim-error-power-connected-strong}
\end{figure}

\section{Software reproduction}\label{app:software-reproduction}

\begin{sloppypar}
The supplementary archive \path{fixed-time-ejs-reproducibility.zip} contains
all code, frozen public inputs, and reference results for the simulations and
TARP analyses. It includes the exact \path{HomoContagioNet} source used for
these results (commit \texttt{d53cd794fcf1588799bd7a764487f1fb610d648b}).
The maintained package is at
\url{https://github.com/duncan-clark/HomoContagioNet}.

R, Git, and the R package \path{jsonlite} are required. The recorded run used
R~4.5.3 and \path{jsonlite}~2.0.0. A full rerun requires macOS or Linux,
uses 16 workers, and took about 29 minutes on a 16-core Apple M4 Max;
runtime varies with hardware.

After unzipping the archive, follow the setup commands in its
\path{README.md}, then choose either option:
\end{sloppypar}
\begin{enumerate}
\item \textbf{Check the supplied results.} Restore the reference results and
run the supplied checker. Its 15 checks verify the recorded source, inputs,
outputs, and expected file structure. This checks the archived run without
recomputing the analyses; it does not establish the statistical assumptions
or verify every number typeset in the paper.
\item \textbf{Rerun the analysis.} From the prepared source checkout, run the
publication driver using an unused working directory outside the checkout:
{\footnotesize
\begin{verbatim}
PUBLICATION_SCRATCH_DIR="$PWD/../reproduction-work" \
  Rscript inst/examples/run_fixed_time_publication.R
\end{verbatim}
}
The driver recomputes all simulations and TARP analyses using the recorded
settings and seeds, then checks and saves the results. No data downloads are
needed.
\end{enumerate}
\begin{sloppypar}
The archive's \path{TECHNICAL_DETAILS.md} records the full settings, output
locations, and restart instructions. \path{VERIFICATION.md} documents the
previous reproduction checks. File hashes and run metadata may differ across
machines even when numerical results agree.
\end{sloppypar}

\section{TARP data construction details}\label{app:tarp-network}

This appendix records application specification details deferred from
Sections~\ref{sec:tarp-network} and~\ref{sec:tarp-application}.

\subsection{Primary cosponsorship network}
Working ties $A_{ij}$ are constructed from legislative collaboration prior to
the first TARP vote. Cosponsorship data come from \cite{congress_data}. An
undirected working tie is formed from House HR bills introduced on or before
the cutoff. For each bill, let $S(b)$ denote the unique sponsor and $C(b)$ the set of
cosponsors mapped to BioGuide IDs whose signing date is on or before the
cutoff. A cosponsorship edge is recorded only
between the sponsor and each cosponsor; two cosponsors of the same bill are not
joined by that bill. The adjacency matrix is
\[
A_{ij}^{\text{cos}}
=
I\Big\{
\exists\, b:\
\bigl(i=S(b)\ \text{and}\ j\in C(b)\bigr)
\text{ or }
\bigl(j=S(b)\ \text{and}\ i\in C(b)\bigr),\
\mathcal{F}(b)=1
\Big\},
\]
where $\mathcal{F}(b)$ applies the bill filters below. Edges are symmetrized and
the diagonal is set to zero. Four pre-specified filters are applied:
\begin{enumerate}
\item \textbf{Chronology first.} A cosponsor event must have a parseable
signing date on or before September 29, 2008.
\item \textbf{Active cosponsorship only.} Cosponsors that withdrew are not counted.
\item \textbf{Exclude commemorative and naming bills.} Bills are removed if
the combined subject-and-title text matches the heuristic patterns listed below.
\item \textbf{Cosponsor cap.} Bills with ten or more cosponsors are excluded.
Very large cosponsorships are less likely to capture a true working relationship
rather than a show of support for a bill.
\end{enumerate}

\subsubsection{Cosponsorship bill filters}
For each bill, the subject terms, top-level subject term, and official, short,
and popular titles are concatenated in that order, with spaces between entries,
and converted to lowercase. A commemorative/naming bill is flagged if this
combined string matches any of the following Perl regular expressions:
\texttt{commemorat}, \texttt{congressional tributes},
\texttt{monuments and memorials}, \texttt{postal service facilities},
\texttt{names legislation}, \texttt{national .\{0,20\} day},
\texttt{sense of (the )?congress}, \texttt{expressing support},
\texttt{honoring}, \texttt{congratulating},
\texttt{to designate the facility}, \texttt{to designate the building},
\texttt{to designate the postal}, \texttt{to rename},
\texttt{to authorize the president to award}, \texttt{gold medal},
\texttt{congressional gold medal}. The patterns \texttt{honoring},
\texttt{congratulating}, and \texttt{to rename} each end with an additional
literal space.

\subsubsection{Bill and edge counts (110th, cutoff 2008-09-29)}
Among HR bills with at least one cosponsor: $4{,}901$ total; $439$ commemorative/naming
excluded; and $2{,}287$ retained after active-only, non-commemorative,
$<10$-cosponsor, and chronology filters. The chronology audit excludes 763
sponsor--cosponsor signing events dated after the cutoff and no events with
missing or malformed signing dates. The primary network has $5{,}778$ unique
undirected edges among $433$ TARP voters.

\begin{figure}[H]
\centering
\includegraphics[width=0.78\linewidth]{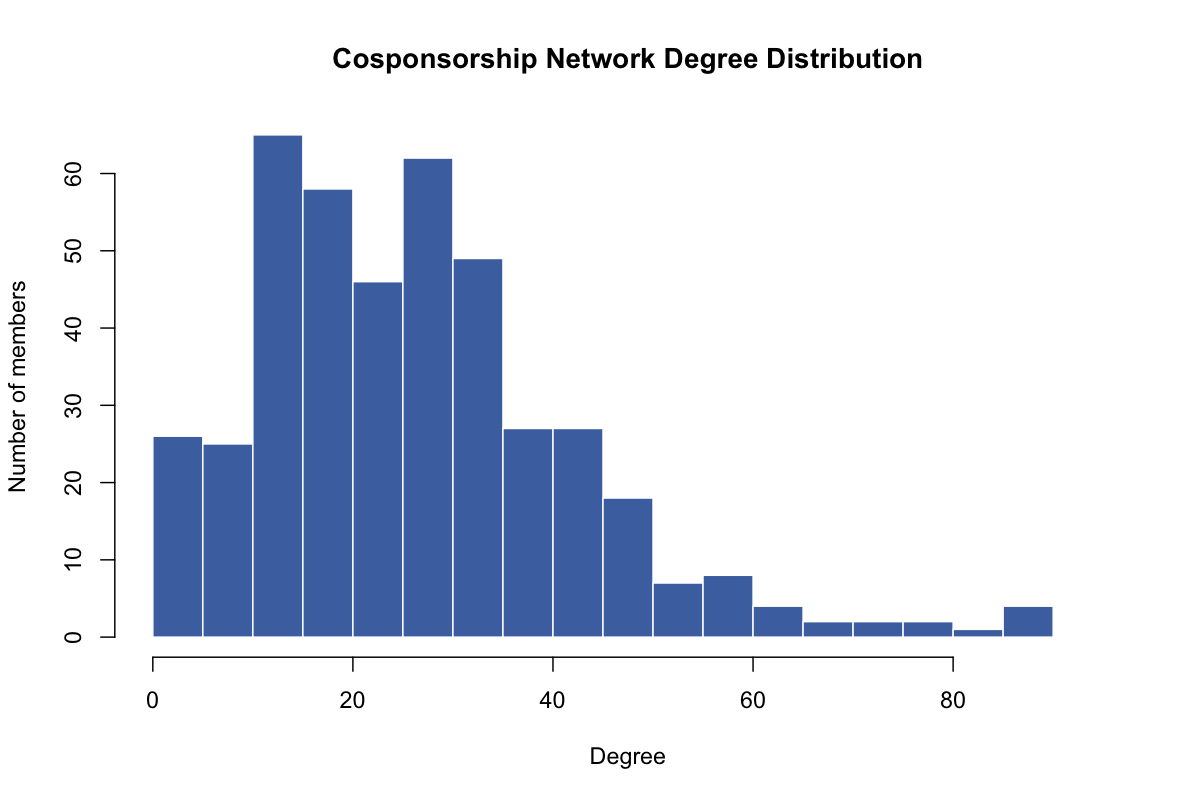}
\caption{Degree distribution for the primary clean cosponsorship network. Ties
are active sponsor--cosponsor pairs signed by the cutoff on
non-commemorative HR bills with fewer than ten cosponsors.}
\label{fig:tarp-degree}
\end{figure}

\section{TARP party-calibration calculations}\label{app:tarp-party-calibration}

All calculations below restrict to switchable dyads, $Y_i^1=0$.
Outcome-risk variation is the range of second-vote Yea risk across ego
party. Connected and selection mix shifts use the party-dyad type
$u\in\{D\!-\!D,D\!-\!R,R\!-\!D,R\!-\!R\}$.

Table~\ref{tab:tarp-calib-outcome} records those ego-party risks:
\[
\widehat R_y^{\text{party}}
=
\frac{\max_{e\in\{D,R\}}\widehat r_y(e)}{\min_{e\in\{D,R\}}\widehat r_y(e)},
\qquad
\widehat r_y(e)
=
\widehat\Pr(Y_i^2=1\mid Y_j^1=y,X_i=e,Y_i^1=0).
\]
Thus
$\widehat R_1^{\text{party}}=0.338/0.219=1.54$ and
$\widehat R_0^{\text{party}}=0.304/0.167=1.82$.

\begin{table}[H]
\centering
\caption{Ego-party outcome-risk calibration among switchable connected
dyads. $\widehat R_y^{\text{party}}$ is the larger ego-party risk over the
smaller, within each alter-vote arm.}
\label{tab:tarp-calib-outcome}
\begin{tabular}{lcccccc}
\hline
 & \multicolumn{3}{c}{Alter Yea ($Y_j^1=1$)}
 & \multicolumn{3}{c}{Alter Nay ($Y_j^1=0$)} \\
Ego party & Yea & $n$ & Risk & Yea & $n$ & Risk \\
\hline
$D$ & 517 & 1{,}529 & 0.338 & 446 & 1{,}465 & 0.304 \\
$R$ & 244 & 1{,}114 & 0.219 & 323 & 1{,}935 & 0.167 \\
\hline
$\widehat R_y^{\text{party}}$ & & & 1.54 & & & 1.82 \\
\hline
\end{tabular}
\end{table}

Table~\ref{tab:tarp-calib-mix} is the connected-population shift. Write
$p_1(u)$, $p_0(u)$, and $p_A(u)$ for the party-dyad shares among Yea-alter,
Nay-alter, and all connected switchable dyads. Then
\[
\widehat R_1^{\text{conn,party}}
=
\max_u\frac{p_1(u)}{p_A(u)},
\qquad
\widehat R_0^{\text{conn,party}}
=
\max_u\frac{p_A(u)}{p_0(u)}.
\]
The maxima are $1.30$ and $1.31$, and both are attained at $R\!-\!D$.

\begin{table}[H]
\centering
\caption{Connected party-dyad mix among switchable dyads. Shares sum to one
within each column. The last two columns are the armwise composition
ratios; the calibration uses the column maxima.}
\label{tab:tarp-calib-mix}
\begin{tabular}{lccccc}
\hline
Party dyad & $p_1$ & $p_A$ & $p_0$ & $p_1/p_A$ & $p_A/p_0$ \\
\hline
$D\!-\!D$ & 0.504 & 0.394 & 0.308 & 1.28 & 1.28 \\
$D\!-\!R$ & 0.074 & 0.101 & 0.123 & 0.73 & 0.83 \\
$R\!-\!D$ & 0.209 & 0.161 & 0.123 & 1.30 & 1.31 \\
$R\!-\!R$ & 0.212 & 0.344 & 0.446 & 0.62 & 0.77 \\
\hline
\end{tabular}
\end{table}

Table~\ref{tab:tarp-calib-sel} is the selection-into-ties shift, comparing
the connected mix $p_A$ with the party-dyad mix $p$ among all ordered
switchable dyads (including non-ties). Then
\[
\widehat R_1^{\text{sel,party}}
=
\max_u\frac{p_A(u)}{p(u)},
\qquad
\widehat R_0^{\text{sel,party}}
=
\max_u\frac{p(u)}{p_A(u)},
\]
equal to $1.75$ at $D\!-\!D$ and $1.98$ at $R\!-\!D$.

\begin{table}[H]
\centering
\caption{Selection-into-ties party-dyad mix among switchable dyads. The
distribution $p$ is the party-dyad share in the full ordered-dyad frame;
$p_A$ is the connected mix from Table~\ref{tab:tarp-calib-mix}.}
\label{tab:tarp-calib-sel}
\begin{tabular}{lcccc}
\hline
Party dyad & $p_A$ & $p$ & $p_A/p$ & $p/p_A$ \\
\hline
$D\!-\!D$ & 0.394 & 0.226 & 1.75 & 0.57 \\
$D\!-\!R$ & 0.101 & 0.191 & 0.53 & 1.88 \\
$R\!-\!D$ & 0.161 & 0.317 & 0.51 & 1.98 \\
$R\!-\!R$ & 0.344 & 0.266 & 1.29 & 0.77 \\
\hline
\end{tabular}
\end{table}

The main text reports only the naturally connected limits
(Table~\ref{tab:tarp-formal-limits}). Table~\ref{tab:tarp-formal-limits-fc}
applies the same party-adjusted contrast and the same $\alpha$ grid to the
forced-contact factor, which includes the selection ratios above.

\begin{table}[H]
\centering
\caption{Party-adjusted one-sided $95\%$ lower limits for the forced-contact
target among switchable egos, primary cosponsorship network. The point bound
is the party-adjusted connected contrast divided by the forced-contact
factor. Mild, moderate, and strong restrictions scale party's measured
strength by $\alpha=1/4$, $1/2$, and $1$. Columns match
Table~\ref{tab:tarp-formal-limits}.}
\label{tab:tarp-formal-limits-fc}
\begin{tabular}{lrrrrr}
\hline
Restriction & $BF$ & Point bound & AS LCL & IE LCL & Ego LCL \\
\hline
Party in $c$, no latent $U$ & $1$ & $1.155$ & $1.039$ & $1.073$ & $1.080$ \\
Mild ($\alpha=1/4$) & $1.073$ & $1.076$ & $0.968$ & $1.000$ & $1.007$ \\
Moderate ($\alpha=1/2$) & $1.236$ & $0.935$ & $0.841$ & $0.868$ & $0.874$ \\
Strong ($\alpha=1$) & $1.724$ & $0.670$ & $0.603$ & $0.623$ & $0.627$ \\
\hline
\end{tabular}
\end{table}

\section{TARP working-network definitions}\label{app:tarp-network-definitions}

The main analysis uses signed, non-commemorative HR cosponsorship with fewer
than ten cosponsors. Table~\ref{tab:tarp-network-definitions} repeats the
party-adjusted switchable-ego calculation, changing only the working network.
Latent-homophily factors are held at the primary-network party calibration:
no latent $U$ ($BF=1$) and the mild restriction ($BF=1.021$).
Cosponsorship variants keep the signed and non-commemorative
filters and vary the per-bill cap over $\{3,5,7,10,15,20\}$. Repeated-collaboration
and early-cosponsor variants use the same primary bill set (signed,
non-commemorative HR bills with fewer than ten cosponsors) and then keep a pair
only if they share at least $k\in\{2,3,5\}$ of those bills, or keep only the
first $N\in\{3,5,10\}$ signers on each such bill. Institutional variants replace
the tie definition by 110th House standing-committee or subcommittee
co-membership.

After party adjustment and with no latent $U$, the actor-sum limit
exceeds one on the primary network, all caps except the sparsest
fewer-than-3 graph, and the early-cosponsor graphs. Under the mild
restriction it remains above one on the primary network, caps of 7, 15, and
20, and the first-3, first-5, and first-10 graphs; the fewer-than-5 endpoint
falls from $1.003$ to $0.982$. All repeated-collaboration and institutional
graphs have limits below one even without latent $U$.

\begin{table}[H]
\centering
\caption{Party-adjusted connected analysis among switchable egos under
alternative working-network definitions. RR is standardized to the
connected party-dyad distribution. AS is the actor-sum one-sided
$95\%$ limit with no latent $U$; the mild column uses the primary-network
party-calibrated factor $BF=1.021$. Deg.\ is mean degree and $n$ is the
number of switchable connected ordered dyads. The starred row is the
primary network.}
\label{tab:tarp-network-definitions}
\begin{tabular}{lrrrrrr}
\hline
Network & Deg. & Isol. & $n$ & RR & AS LCL & Mild \\
\hline
\multicolumn{7}{l}{\emph{Cosponsorship cap}} \\
Fewer than 3 cosponsors & $5.1$ & $23$ & $1{,}152$ & $1.139$ & $0.859$ & $0.841$ \\
Fewer than 5 cosponsors & $10.8$ & $6$ & $2{,}432$ & $1.184$ & $1.003$ & $0.982$ \\
Fewer than 7 cosponsors & $17.2$ & $2$ & $3{,}882$ & $1.225$ & $1.072$ & $1.050$ \\
Fewer than 10 cosponsors$^{*}$ & $26.7$ & $1$ & $6{,}043$ & $1.155$ & $1.039$ & $1.018$ \\
Fewer than 15 cosponsors & $41.3$ & $0$ & $9{,}159$ & $1.172$ & $1.077$ & $1.055$ \\
Fewer than 20 cosponsors & $55.3$ & $0$ & $12{,}178$ & $1.121$ & $1.045$ & $1.024$ \\
\hline
\multicolumn{7}{l}{\emph{Repeated collaboration}} \\
At least 2 shared bills & $5.4$ & $31$ & $1{,}247$ & $0.986$ & $0.770$ & $0.754$ \\
At least 3 shared bills & $1.8$ & $144$ & $412$ & $1.355$ & $0.847$ & $0.830$ \\
At least 5 shared bills & $0.4$ & $322$ & $96$ & $0.867$ & $0.314$ & $0.308$ \\
\hline
\multicolumn{7}{l}{\emph{Early cosponsors}} \\
First 3 cosponsors & $17.3$ & $1$ & $3{,}928$ & $1.184$ & $1.031$ & $1.009$ \\
First 5 cosponsors & $22.6$ & $1$ & $5{,}126$ & $1.170$ & $1.036$ & $1.015$ \\
First 10 cosponsors & $26.7$ & $1$ & $6{,}043$ & $1.155$ & $1.039$ & $1.018$ \\
\hline
\multicolumn{7}{l}{\emph{Institutional working ties}} \\
Standing-committee co-membership & $87.9$ & $11$ & $20{,}735$ & $1.051$ & $0.978$ & $0.958$ \\
Subcommittee co-membership & $56.7$ & $24$ & $13{,}560$ & $1.025$ & $0.949$ & $0.929$ \\
\hline
\end{tabular}
\end{table}

\begin{supplement}
\stitle{Code, data, and computational reproduction materials}
\sdescription{The archive \path{fixed-time-ejs-reproducibility.zip}
contains the exact package source revision, frozen public legislative inputs,
reference simulation and TARP results, environment and provenance records,
and instructions for verification and full reproduction.
Appendix~\ref{app:software-reproduction} describes its use.}
\end{supplement}

\bibliographystyle{imsart-number}
\bibliography{bib}

\end{document}